\documentclass[runningheads,12pt]{llncs}
\usepackage{graphicx}
\usepackage[labelformat=simple]{subcaption}

\usepackage{floatrow} 
\usepackage{amsmath,epstopdf}
\usepackage{enumerate}
\usepackage{enumitem}
\usepackage{amssymb}
\usepackage{setspace}
\usepackage[linesnumbered, ruled, vlined, noend]{algorithm2e}
\usepackage{color}
\usepackage[]{color}
\usepackage[]{url}
\usepackage[colorlinks,linkcolor=black,urlcolor=black,citecolor=black]{hyperref}

\usepackage{cleveref}
\newenvironment{appendix-lemma}[1]{\vspace{0.1in}\noindent{\bf Lemma~#1~} \em }{\vspace{0.1in}}
\newenvironment{appendix-claim}[1]{\vspace{0.1in}\noindent{\bf Claim~#1~} \em }{\vspace{0.1in}}
\newenvironment{appendix-theo}[1]{\vspace{0.1in}\noindent{\bf Theorem~#1~} \em }{\vspace{0.1in}}
\usepackage{tikz}
\usepackage{standalone}
\usetikzlibrary{arrows}
\usetikzlibrary{decorations,fit,shapes} 
\usetikzlibrary{decorations.markings}
\usetikzlibrary{arrows.meta,positioning}
\usetikzlibrary{shapes.misc}
\tikzset{cross/.style={cross out, draw=black, minimum size=2*(#1-\pgflinewidth), inner sep=0pt, outer sep=0pt},
	cross/.default={1pt}}
\usetikzlibrary{arrows, decorations.pathmorphing}

\tikzstyle{vertex}=[auto=left,circle,draw=black!80,fill=none,minimum size=15pt,inner sep=0pt]

\tikzset{
    photon/.style={decorate, decoration={snake}, draw=red}}

\newcommand{\etal}{\textit{et al}. }

\usepackage[margin = 1.05in]{geometry}

\begin{document}
\title{Optimal Capacity Modification for Stable Matchings with Ties\thanks{A preliminary version~\cite{RanjanNN25} of this work appeared in the proceedings of the 34th International Joint Conference on Artificial Intelligence (IJCAI) 2025}} 
%

\titlerunning{Stable Matchings with Ties} 

\author{Keshav Ranjan\inst{1}\,\thanks{This work was done while the author was at IIT Madras.} \and Meghana Nasre\inst{2} \and
Prajakta Nimbhorkar\inst{3}}
%
%
\authorrunning{K. Ranjan \etal } 
\institute{Chennai Mathematical Institute, Chennai, Tamil Nadu, India \and IIT Madras, Chennai, Tamil Nadu, India \and
Chennai Mathematical Institute and UMI ReLaX, Chennai, Tamil Nadu, India}
%

\newcommand{\HR}{\mbox{{\sf HR}}}
\newcommand{\prefa}{\mbox{{\sf Pref($a$)}}}
\newcommand{\prefap}{\mbox{{\sf Pref($a_1$)}}}
\newcommand{\prefal}{\mbox{{\sf List($a$)}}}
\newcommand{\preflqa}{\mbox{{\sf PrefC($a$)}}}
\newcommand{\prefb}{\mbox{{\sf Pref($b$)}}}
\newcommand{\prefh}{\mbox{{\sf Pref($h$)}}}
\newcommand{\prefhp}{\mbox{{\sf Pref($h'$)}}}
\newcommand{\prefr}{\mbox{{\sf Pref($r$)}}}
\newcommand{\prefrp}{\mbox{{\sf Pref($r'$)}}}
\newcommand{\prefrpp}{\mbox{{\sf Pref($r''$)}}}
\newcommand{\prefrT}{\mbox{{\sf Pref($r_2$)}}}
\newcommand{\prefbi}{\mbox{{\sf Pref($b_i^{i_3}$)}}}

\newcommand{\prefu}{\mbox{{\sf Pref($u$)}}}
\newcommand{\rankv}[1]{\mbox{{\sf $rank_{v}(#1)$}}}
\newcommand{\ranku}[1]{\mbox{{\sf $rank_{u}(#1)$}}}
\newcommand{\ranka}[1]{\mbox{{\sf $rank_{a}(#1)$}}}
\newcommand{\rankb}[1]{\mbox{{\sf $rank_{b}(#1)$}}}
\newcommand{\prefv}{\mbox{{\sf Pref($v$)}}}
\newcommand{\HRLQ}{\mbox{{\sf HRLQ}}}
\newcommand{\SMLQ}{\mbox{{\sf SM-1C}}}
\newcommand{\RSM}{\mbox{{\sf RSM}}}
\newcommand{\RWSM}{\mbox{{\sf RSM}}}
\newcommand{\SMTI}{\mbox{{\sf SMTI}}}
\newcommand{\HRTTLQ}{\mbox{{\sf HR-2T-2LQ}}}
\newcommand{\LQ}{\mbox{{\sf LQ}}}

\newcommand{\CL}{\mbox{{\sf CL}}}
\newcommand{\Df}[1]{\mbox{{\sf $def(#1)$}}}
\newcommand{\Dfa}[1]{\mbox{{\sf $def_{\mathcal{A}}(#1)$}}}
\newcommand{\Dfb}[1]{\mbox{{\sf $def_{\mathcal{B}}(#1)$}}}
\newcommand{\Dfh}[1]{\mbox{{\sf $def_{\mathcal{H}}(#1)$}}}
\newcommand{\Dfr}[1]{\mbox{{\sf $def_{\mathcal{R}}(#1)$}}}
\newcommand{\MAXEFM}{\mbox{{\sf MAXEFM}}}
\newcommand{\MAXSMTI}{\mbox{{\sf MAX SMTI}}}
\newcommand{\CLIQUE}{\mbox{{\sf CLIQUE}}}
\newcommand{\minBP}{\mbox{{\sf MinBP}}}
\newcommand{\minBR}{\mbox{{\sf MinBR}}}
\newcommand{\ESDA}{\mbox{{\sf ESDA}}}
\newcommand{\MSDA}{\mbox{{\sf MSDA}}}
\newcommand{\CSM}{\mbox{{\sf CSM}}}
\newcommand{\EFHRLQ}{\mbox{{\sf{EF}\text{-}\sf{HR}\text{-}\sf{LQ}}}}
\newcommand{\HH}{\mathcal{H}}
\newcommand{\RR}{\mathcal{R}}
\newcommand{\RRM}{\mathcal{R}_m}
\newcommand{\RRMU}{\mathcal{R}_m^u}
\renewcommand{\AA}{\mathcal{A}}
\newcommand{\BB}{\mathcal{B}}
\newcommand{\CC}{\mathcal{C}}
\newcommand{\D}{\mathcal{D}}
\newcommand{\EE}{\mathcal{E}}
\newcommand{\clr}{\textcolor{blue}}
\newcommand{\rvrO}{\textcolor{blue}}
\newcommand{\rvrT}{\textcolor{orange}}
\newcommand{\rvrTh}{\textcolor{cyan}}
\newcommand{\clrR}{\textcolor{red}}
\newcommand{\clrg}{\textcolor{lightgray}}
\renewcommand{\S}{{s}}
\newcommand{\T}{{t}}
\newcommand{\MINSUM}{\mbox{{\sf \textsc{MinSum}}}}
\newcommand{\MINSUMR}{\mbox{{\sf \textsc{MinSum-R}}}}
\newcommand{\MINMAX}{\mbox{{\sf \textsc{MinMax}}}}
\newcommand{\MINSUMSP}{\mbox{{\sf \textsc{MinSumSP}}}}
\newcommand{\MINMAXSP}{\mbox{{\sf \textsc{MinMaxSP}}}}
\newcommand{\MINSUMSSP}{{\sf \textsc{MinSumSSP}}}
\newcommand{\MINSUMSS}{{\sf \textsc{MinSum-SS}}}
\newcommand{\MINCOST}{\mbox{{\sf \textsc{MinCost}}}}

\newcommand{\MINMAXSSMP}{{\sf \textsc{MinMaxSSM} Problem}}
\newcommand{\CAPSSMP}{{\sf \textsc{1-or-2 Capacity-SS problem}}}
\newcommand{\RPCAPSSMP}{{\sf 1-or-2 Capacity \textsc{SSP} problem}}
\newcommand{\MINMAXSS}{{\sf \textsc{MinMax-SS}}}
\newcommand{\RPMINMAXSS}{{\sf \textsc{MinMax-SSP}}}
\newcommand{\MINMAXSSBT}{{\sf \textsc{MinMax-SS-BT}}}
\newcommand{\MINSUMSSFP}{{\sf \textsc{MinSum-SS-FP}}}
\newcommand{\MINSUMSSFE}{{\sf \textsc{MinSum-SS-FE}}}
\newcommand{\MINSUMCOST}{{\sf \textsc{MinSum-Cost}}}
\newcommand{\MINMAXSSRP}{\mbox{{\sf \textsc{MinMaxSSRP }}}}
\newcommand{\SMI}{\mbox{{\sf SMI}}}
\newcommand{\HRHT}{\mbox{{\sf HR-HT}}}
\newcommand{\HRT}{\mbox{{\sf HRT}}}
\newcommand{\AugSSM}{{\sf 1-or-2 Capacity SSM}}
\newcommand{\SAT}{\mbox{{\sf 1-in-3 SAT}}}
\newcommand{\MSAT}{{{\sf \textsc{Monotone 1-in-3 SAT}}}}
\newcommand{\MNESAT}{{\sf \textsc{Monotone Not-All-Equal 3-SAT}}}
\newcommand{\SSM}{{strongly stable matching}}
\newcommand{\SBP}{{strong blocking pair}}
\newcommand{\SBPM}{{strong blocking pair with respect to the matching $M$}}
\newcommand{\keshav}[1]{{\leavevmode\color{red}{\ (Keshav: #1)}}}

\newcommand{\ct}[1]{\mbox{{$C_{#1}$}}}
\newcommand{\gt}[1]{\mbox{{$G_{#1}$}}}
\newcommand{\pref}[1]{\sf Pref(#1)}
\newcommand{\II}{\mathcal{I}}

\newcommand{\tG}{\widetilde{G}}
\newcommand{\br}[1]{\overline{#1}}
\newcommand{\tth}{\widetilde{h}}
\newcommand{\tM}{\widetilde{M}}
\newcommand{\tq}{\widetilde{q}}
\newcommand{\tHH}{\tilde{\mathcal{H}}}
\newcommand{\tL}{\mathcal{L'}}
\newcommand{\tLR}{\tilde{\mathcal{L}_r}}
\newcommand{\tLH}{\tilde{\mathcal{L}_h}}
\newcommand{\tE}{\tilde{{E}}}

\newcommand{\pn}[1]{{\leavevmode\color{teal}{\ (pn: #1)}}}

\newcounter{dummy} \numberwithin{dummy}{section}
\newtheorem{thm}[dummy]{Theorem}
\newtheorem{defi}[dummy]{Definition}
\newtheorem{lem}[dummy]{Lemma}
\newtheorem{cl}[dummy]{Claim}
\newtheorem{obs}[dummy]{Observation}

\maketitle              
\begin{abstract}
We consider the Hospitals/Residents (\HR) problem in the presence of ties in hospital preferences. Among the three notions of stability, namely weak stability, strong stability, and super-stability, we focus on the notion of strong stability. Strong stability has many desirable properties, both theoretically and in practice; however, its existence is not guaranteed.
In this paper, our objective is to optimally increase the quotas of hospitals to ensure that a strongly stable matching exists in the modified instance. 

We explore two natural optimization criteria: (i) minimizing the total capacity increase across all hospitals (\MINSUM) and (ii) minimizing the maximum capacity increase for any hospital (\MINMAX). We show that the \MINSUM\ problem admits a polynomial-time algorithm. 
We also establish an analog of the well-known rural hospitals theorem [Gale \& Sotomayor, 1985; Roth, 1986], adapted to the \MINSUM\ augmentation setting. 

We consider a generalization of the \MINSUM\ problem in which each hospital incurs a cost per
unit increase in its quota. We show that the cost version of the \MINSUM\ problem is NP-hard and inapproximable within any multiplicative factor, even if the costs are zero or one. 
For the \MINSUM\ objective with a set of {\em forced }edges, we give a polynomial-time algorithm.

In contrast to the above results for the \MINSUM\ problem, we show that the \MINMAX\ problem is NP-hard. 
When hospital preference lists have ties of length at most $\ell+1$, we give a polynomial-time algorithm that increases each hospital’s quota by at most $\ell$, ensuring the resulting instance admits a strongly stable matching. Moreover, among all such augmentations, 
our algorithm outputs the best strongly stable matching from the residents' perspective. 

\end{abstract}

\section{Introduction}
The Hospitals/Residents (\HR) problem~\cite{GusfieldIrvingBook,RothSotomayor1990Book} is a many-to-one generalization of the classical stable marriage problem~\cite{GS62}. As the name suggests, the \HR\ problem models the assignment of junior doctors (residents)
to hospitals where agents in both sets are allowed to rank acceptable agents from the other set in a preference order. The problem has been extensively investigated, as it has applications in a number of centralized matching schemes in many countries, including the National Resident Matching Program in the USA ({\sf NRMP}), the Canadian Resident Matching Service ({\sf CaRMS}), and the Scottish Foundation Allocation Scheme ({\sf SFAS}), to name a few. In addition, the \HR\ problem models several real-world applications, such as assigning children to schools~\cite{abdulkadirouglu2005new} and students to undergraduate programs~\cite{BaswanaCCKP19}, where agents need to be matched to programs, and both sets express preferences over each other. 

We consider a generalization of the \HR\ problem where {\em ties} are allowed in the preference lists. That is, agents can be indifferent between multiple agents from the other set. This problem is known as the Hospitals/Residents problem with ties (\HRT). 
Ties in preference lists play an important role in real-world matching applications. For instance, hospitals with many applicants often find it difficult to generate
strict preference lists and, within the framework of a centralized matching scheme, have expressed a desire to include ties in their preference lists~\cite{swat/IrvingMS00}. In the context of college admissions, it is natural for colleges to place all students with equal scores in a single tie in their preference lists.

The classical notion of stability defined for strict preferences has been generalized in the literature for the case of ties, in three different ways --- weak stability, strong stability and super-stability (see Definition~\ref{def:strong} and the footnote therein). As indicated by the names, super-stability is the strongest notion, and weak stability is the weakest among the three. It is well-known that every instance of the \HRT\ problem admits a weakly stable matching (and it can be obtained by breaking ties arbitrarily and computing a stable matching in the resulting strict-list instance); however, strong or super-stable matchings are not guaranteed to exist~\cite{irving1994stable}.

The strongest notion of stability is super-stability. However, as noted in~\cite{irving2003strong}, insisting on super-stability in practical scenarios can be overly restrictive and is less likely to be attainable. Moreover, in applications such as college admissions, students are naturally expected to
express strict preferences over colleges, even though colleges may place students with equal scores in a single tie\footnote{We refer to this as \HRHT\ in this paper.}. In such scenarios, super and strong stability coincide. On the other hand, weak stability is too weak, and as justified in~\cite{Manlove02Lattice}, it is susceptible to compromise through persuasion or bribery (also see~\cite{irving2003strong,kunysz2016characterisation} for further details). Moreover, from a social perspective, weak stability may not be an acceptable notion despite its guaranteed existence. For instance, according to the equal treatment policy used in Chile and Hungary, it is not acceptable that a student is rejected from a college preferred by her, even though other students with the same score are admitted (see~\cite{CsehH20Restricted} and the references therein). Thus, strong stability is not only appealing but also essential.

Since strong stability is desirable but not always guaranteed to exist, a natural approach is to increase the quotas of hospitals (or colleges) so that the modified instance admits a strongly stable
matching. We address this problem in this paper. We use the hospital-residents terminology, as is customary in many-to-one stable matchings. Thus, we seek to increase or {\em augment }the hospital quotas to obtain a modified instance that admits a strongly stable matching. We note that, unlike the stable matching problem with strict preference lists\textemdash which is monotone with respect to capacity increases\textemdash the strongly stable matching problem with ties is not monotone with respect to capacity augmentation. That is, an instance that admits a strongly stable matching may stop admitting one after the capacities of certain hospitals are increased. We illustrate this with the following example. 

 Consider an instance $G=(\RR\cup\HH, E)$ shown in Fig.~\ref{fig:NonMonotonic} where $\RR=\{r_1, r_2,r_3,r_4\}$, $\HH=\{h\}$ and the quota of the hospital $h$ is $q(h)=1$. The instance does not admit a \SSM, since either $r_1$ or $r_2$ remains unmatched and forms a blocking pair with $h$. If we set $q(h)=2$, then $\{(r_1,h),(r_2,h)\}$ is a \SSM\ in the resulting instance. However, if we further augment the quota of $h$ and set $q(h)=3$, then the resulting instance does not admit a \SSM, since at least one resident must remain unmatched in any matching and forms a blocking pair with $h$.


\begin{figure}
\begin{tabular}{ll p{0.3 cm}  p{0.3 cm} ll}
			$r_1$ :& $h$ & & & &\\[2pt]
			$r_2$ :& $h$ & & & &\\[2pt]
                $r_3$ :& $h$ & & & &\\[2pt]
                $r_4$ :& $h$& & & $[1]\ h$ :& $(r_1,r_2), (r_3,r_4)$\\[2pt]
		\end{tabular}
            \caption{An instance $G$ illustrating the non-monotonicity of the \SSM\ problem.}
        \label{fig:NonMonotonic}
\end{figure}

Therefore, an instance that admits a strongly stable matching may cease to admit one after the capacities of certain hospitals are increased. This non-monotonicity brings new challenges.

We explore two natural optimization criteria: (i) minimizing the total increase (sum) in quotas across all hospitals (\MINSUM), and (ii) minimizing the maximum increase in the quota for any hospital (\MINMAX).
Our work falls under the broad theme of capacity planning/modification, which has received significant attention~\cite{ChenCsaji23CapAAMAS,GokhaleSNV24AAMASCap,BobbioCLT22CapacityVar,BobbioCLRT23Capacity,AfacanDL24CapDesign}, motivated by practical
applications in which quotas are not rigid. To the best of our knowledge, our work is the first one to explore capacity augmentation for the 
notion of strong stability, although capacity modification has received considerable attention in the strict list setting (see Section~\ref{subsec:relWork}).

\subsection{Preliminaries and notation}
The input to our problem is a bipartite graph $G=(\RR\cup\HH, E)$, where the vertex set $\RR$ represents the set of residents, $\HH$ represents the set of hospitals and the edge set $E$ represents mutually acceptable resident-hospital pairs. We set $n = |\RR| + |\HH|$ and $m = |E|$.
Each hospital $h\in \HH$ has an associated quota $q(h)$, denoting the maximum number of residents that may be assigned to it in any valid matching.
Each vertex $v\in\RR\cup\HH$ ranks its neighbors according to its preference ordering, called {\em the preference list of $v$} and denoted \prefv.
We say that a vertex strictly prefers a neighbor with a smaller rank over another neighbor with a larger rank. If a vertex is allowed to be indifferent between some of its neighbors and is allowed to assign the same rank to such neighbors, it is referred to as a {\em tie}. The length of a tie is the number of neighbors having equal rank. If ties are not allowed (or equivalently, all ties have length 1), the preference lists are said to be {\em strict}. We use $u_1\succ_v u_2$ to denote that $v$ strictly prefers $u_1$ over $u_2$ and 
 $u_1\succeq_v u_2$ to denote that $v$ either strictly prefers $u_1$ over $u_2$ or is indifferent between them.

A matching $M$ in $G$ is a subset of $E$ such that for each resident $r\in\RR$ we have $|M(r)|\le 1$ and for each hospital $h\in\HH$ we have $|M(h)|\le q(h)$ where $M(v)$ denotes the set of matched partners of $v$ in $M$. 
For a resident $r$, if $|M(r)|=0$, then $r$ is unmatched in $M$. In this case, we denote the matched partner of $r$ by $M(r)=\bot$. 
A hospital $h\in\HH$ is said to be fully subscribed in $M$ with respect to its quota $q(h)$ if $|M(h)| = q(h)$, and under-subscribed in $M$ if $|M(h)| < q(h)$. By a slight abuse of terminology, we say that $h$ is over-subscribed in $M$ if $|M(h)| > q(h)$. If left unspecified, the quota under consideration for these terms is the original quota $q(h)$. If $h$ is under-subscribed, then we implicitly match the remaining $q(h)-|M(h)|$ many positions of $h$ to as many copies of $\bot$. A vertex prefers any of its neighbors over $\bot$. 

\begin{defi}[\bf Strong stability]\label{def:strong} For a matching $M$, an edge  $(r, h)\in E\setminus M$ is a \SBP\ with respect to $M$, if either (i) or (ii) holds:
    \begin{itemize} 
    \item[(i)]  $h \succ_r M(r)$ and  there exists  $r' \in M(h)$ such that $ r \succeq_h r'$ 
    \item[(ii)] 
    $h \succeq_r M(r)$ and there exists   $r' \in M(h) $ such that $r \succ_h r'$. 
    \end{itemize}
     A matching $M$ is a \SSM\ if there does not exist any \SBP\ with respect to $M$. \footnote{A pair $(r,h)$ is a {\em super blocking pair} if both prefer each other strictly or equally to their matched partners. Also, $(r,h)$ forms a {\em weak blocking pair} if they prefer each other strictly more than their matched partners.}
\end{defi}

Throughout the paper, we refer to a strong blocking pair as a blocking pair. 
We give a simple example to illustrate that a strongly stable matching is not guaranteed to exist. Consider an instance with one resident $r$ and two hospitals $h_1,h_2$, where $r$ has $h_1$ and $h_2$ tied at rank 1, whereas $h_1,h_2$ have a unit quota each and both of them rank $r$ as a rank-1 vertex. No matching in this instance is strongly stable, since the matching $M_1 = \{(r,h_1)\}$ is blocked by $(r, h_2)$ and $M_2=\{(r,h_2)\}$ is blocked by $(r, h_1)$. 
\begin{center}
\begin{tabular}{ll l l  |  ll l l}
 $r$ & : & $(h_1,h_2)$ & \ \ \ &\ \ \ $[1]\ h_1$ & : & $r$\\ 
 &  &  & \ \ \  &\ \ \ $[1]\ h_2$ & : & $r$
\end{tabular}
\end{center}
Moreover, the same example illustrates that increasing hospital quotas alone may not help in obtaining an instance that admits a \SSM. This happens because there are ties in residents' preference lists, whereas quota augmentation is possible for hospitals only.

What if resident preference lists are strict and ties appear only in hospitals' preferences? We call such instances \HRHT\  (Hospitals/Residents problem with ties on hospitals' side only). There exist simple instances of  \HRHT\  that do not admit a \SSM; however, for any such instance, we can construct an {\em augmented instance} $G'$  by setting the quota of each hospital $h$ equal to its degree in $G$. 
It is easy to observe that the matching $M'$ that assigns each resident to its rank-1 hospital, is a \SSM\ in $G'$. Thus, unlike the general \HRT\ case, an \HRHT\ instance can always be augmented so that the instance admits a \SSM. Our objective in this paper is to optimally increase hospitals' quotas to ensure that a \SSM\ exists in the modified instance.  We are ready to formally define our problems. 

\subsection{Our problems and contributions}\label{sec:problems}
Throughout the paper, unless stated otherwise, we assume that the given \HRHT\ instance $G=(\RR\cup\HH,E)$ does not admit a \SSM. Deciding whether an \HRT\ instance admits a strongly stable matching can be done in polynomial time using the algorithm by Irving \etal~\cite{irving2003strong} (see Appendix~\ref{sec:irving}). Throughout this paper, the {\em augmented} instance $G'$ is the same as $G$ except that $q'(h) \ge q(h)$ for each $h$. Our first objective is to minimize the total increase in quotas across all hospitals, which we denote as the \hyperref[prob:minsumss]{\MINSUMSS} problem.

\vspace{0.1in}

\phantomsection
\label{prob:minsumss}
\noindent\textbf{\underline\MINSUMSS:}  Given an  \HRHT\ instance $G=(\RR\cup\HH,E)$, construct an augmented instance $G'$ 
such that $G'$ admits a \SSM\ and the sum of the increase in quotas over all hospitals (that is, $\sum_{h\in\HH} (q'(h)-q(h))$) is minimized.

\begin{thm}\label{thm:MinSum}
     The \hyperref[prob:minsumss]{\MINSUMSS} problem is solvable in polynomial time.
\end{thm}

We establish an analog of the well-known rural hospitals theorem~\cite{galeSotomayor1985,roth1986allocation} for the \hyperref[prob:minsumss]{\MINSUMSS} problem. We show that all strongly stable matchings across optimal augmented instances match the same set of residents. Furthermore, we prove that under-subscribed hospitals 
are filled to the same extent in every \SSM\ across all optimal augmented instances. These results are stated in Section~\ref{sec:MinSumSSM} as Theorem~\ref{theo:MinSumProp} and Corollary~\ref{cor:fullHospitals}. 

Given the polynomial-time solution to the \hyperref[prob:minsumss]{\MINSUMSS} problem, we consider the optimal total quota augmentation (if possible) required for matching a specified subset of edges $Q\subseteq E$ in $G$.  We denote this problem as the \hyperref[prob:minsumssfe]{\MINSUMSSFE} (forced edges). 
The \hyperref[prob:minsumssfe]{\MINSUMSSFE} problem is motivated by course allocation scenarios, such as assigning final-year university students to specific courses \textit{required} for graduation. These mandatory student-course pairs can be naturally modeled as forced edges. We formally define the problem below. 

\vspace{0.1in}

\phantomsection
\label{prob:minsumssfe}
\noindent\textbf{\underline\MINSUMSSFE:}
 Given an \HRHT\ instance $G = (\RR \cup \HH, E)$, which possibly admits a \SSM, and a subset of forced edges $Q \subseteq E$, construct an augmented instance $G'$, if possible, such that $G'$ admits a \SSM\ that includes all edges in $Q$, and the total increase in hospital quotas, defined as $\sum_{h \in \HH} (q'(h) - q(h))$, is minimized.

\begin{thm}
\label{thm:ForcedEdges}
The \hyperref[prob:minsumssfe]{\MINSUMSSFE} problem is solvable in polynomial time.
\end{thm}

 Next, we consider a generalization of the \hyperref[prob:minsumss]{\MINSUMSS} problem, where increasing the quota of a hospital incurs a cost per unit increase in the quota. This is denoted as the \hyperref[prob:minsumcost]{\MINSUMCOST} problem.
 
\vspace{0.1in}

\phantomsection
\label{prob:minsumcost}
\noindent\textbf{\underline\MINSUMCOST:} 
Given an \HRHT\ instance $G = (\RR \cup \HH, E)$, where each hospital $h \in \HH$ has an associated cost $c(h)$ for each unit increase in its quota, construct an augmented instance $G'$ such that $G'$ admits a \SSM, and the total cost of quota augmentation, defined as $\sum_{h \in \HH} c(h) \cdot (q'(h) - q(h))$, is minimized.

In contrast to the polynomial-time solvability of the \hyperref[prob:minsumss]{\MINSUMSS} problem, we establish the following hardness result for the \hyperref[prob:minsumcost]{\MINSUMCOST} problem.

\begin{thm}
\label{thm:minsumcost}
The \hyperref[prob:minsumcost]{\MINSUMCOST} problem is {\sf NP}-hard, even when both resident preferences and hospital preferences are derived from respective master lists. Moreover, the problem is inapproximable within any multiplicative factor.
\end{thm}

We now turn our attention to the alternative objective\textemdash minimizing the maximum increase in quota for any hospital\textemdash and define the \hyperref[prob:minmaxss]{\MINMAXSS} problem.

\vspace{0.1in}

\phantomsection
\label{prob:minmaxss}
\noindent\textbf{\underline\MINMAXSS:} 
Given an \HRHT\ instance $G=(\RR\cup\HH, E)$, construct an augmented instance $G'$ 
such that $G'$ admits a \SSM, and the maximum increase in the quota for any hospital is minimized, that is, $\max_{h\in\HH}\{q'(h)-q(h)\}$ is minimized.

Our result for the \hyperref[prob:minmaxss]{\MINMAXSS} problem is stated in Theorem~\ref{thm:minmax}.

\begin{thm}
\label{thm:minmax}
The \hyperref[prob:minmaxss]{\MINMAXSS} problem is {\sf NP}-hard even when resident preferences are single-peaked and hospital preferences are derived from a master list. Moreover, the same minimization objective with the goal of constructing an instance that admits a resident-perfect \SSM\ (one that matches all residents) is also {\sf NP}-hard.
\end{thm}

Finally, we consider a variant of the \MINMAX\ problem where hospitals' preference lists have bounded tie lengths. We refer to this as the \hyperref[prob:minmaxssbt]{\MINMAXSSBT} problem.

\phantomsection
\label{prob:minmaxssbt}
\noindent\textbf{\underline\MINMAXSSBT:}
Given an \HRHT\ instance $G=(\RR\cup\HH, E)$, where the length of ties in the preference lists of hospitals is bounded by $\ell + 1$, determine the existence of an augmented instance $G'$ such that $G'$ admits a \SSM, and $\max_{h\in\HH}\{q'(h)-q(h)\}\le \ell$.

We use the term \textit{$\ell$-augmented instance} to denote an augmented instance that admits a \SSM\ and is obtained from $G$ by at most $\ell$ augmentations per hospital, that is, for every hospital $h$ we have $q'(h)-q(h) \le \ell$. An $\ell$-augmented instance $G'$ is a \textit{resident-optimal $\ell$-augmented instance} of $G$ if the resident-optimal\footnote{A \SSM\ $M$ is \textit{resident-optimal} if for each resident $r\in \RR$, $M(r)$ is the best possible hospital to which $r$ can get matched in any \SSM.} \SSM\ in $G'$ is the \textit{best} for residents across \textit{all} $\ell$-augmented instances of $G$.

\begin{thm}\label{thm:MINMAXBT}
    For an instance of the \hyperref[prob:minmaxssbt]{\MINMAXSSBT} problem, where the tie length is at most $\ell + 1$, an $\ell$-augmented instance exists. Moreover, a resident-optimal $\ell$-augmented instance can be computed in polynomial time. Therefore, a \SSM\ that matches the maximum number of residents across all $\ell$-augmented instances can be computed in polynomial time.
\end{thm}

\subsection{Related Work}\label{subsec:relWork}
\noindent {\bf Capacity Modification:} Chen and Cs{\'{a}}ji~\cite{ChenCsaji23CapAAMAS} studied a problem similar to ours in the setting of strict preference lists, with the goal of augmenting hospital quotas so that the resulting instance admits a resident-perfect stable matching. They showed that, under the \MINMAX\ objective, the problem admits a polynomial-time algorithm. In contrast, perhaps surprisingly, the \MINMAX\ objective for \SSM\ yields an NP-hardness result (Theorem~\ref{thm:minmax}). They also consider the \MINSUM\ objective and show that the problem of obtaining an optimal augmented instance that admits a stable and resident-perfect matching under \MINSUM\ is NP-hard. This also implies that constructing an augmented instance in the \HRHT\ setting that admits a strongly stable, resident-perfect matching is NP-hard. However, without the resident-perfectness requirement, our result in Theorem~\ref{thm:MinSum} gives a polynomial-time algorithm.

Capacity modification to achieve specific objectives has attracted significant interest in recent years. Bobbio~\etal~\cite{BobbioCLT22CapacityVar} explored the complexity of determining the optimal variation (augmentation or reduction) of hospital quotas to achieve the best outcomes for residents, subject to stability and capacity variation constraints, and showed NP-hardness results. In a follow-up work, Bobbio~\etal\cite{BobbioCLRT23Capacity} developed a mixed integer linear program to address this issue, and provided a comprehensive set of tools for obtaining near-optimal solutions. 
Gokhale~\etal~\cite{GokhaleSNV24AAMASCap} considered the problem of modifying hospitals' quotas to achieve two objectives -- (i) to obtain a stable matching so as to match a given pair, and (ii) to stabilize a given matching, either by only augmenting or only reducing hospital quotas.
Afacan~\etal\cite{AfacanDL24CapDesign} examined capacity design in the \HR\ setting, to achieve a stable matching that is not Pareto-dominated by any other stable matching.

Kavitha and Nasre~\cite{KavithaN11PopularVar}  and Kavitha~\etal~\cite{KNP12} addressed the capacity augmentation problem for {\em popular} matchings in the one-sided preference list setting (where every hospital is indifferent between its neighbors). It is known that a popular matching is not guaranteed to exist in this setting. Therefore, their objective was to optimally increase hospital quotas to create an instance that admits a popular matching. Although we consider a different setting (two-sided preference lists) and a different optimality notion\textemdash namely, strong stability\textemdash it is noteworthy that our results closely resemble those of Kavitha and Nasre~\cite{KavithaN11PopularVar} and Kavitha~\etal~\cite{KNP12}.

\noindent {\bf Strong Stability:} The notion of strong stability was first studied in the one-to-one setting (\textit{i.e. $q(h)=1$ for all $h\in \HH$}) for balanced, complete bipartite graphs by Irving~\cite{irving1994stable}, where he gave an $O(n^4)$ algorithm to compute a \SSM\ if it exists. Since then, the \SSM\ problem has received significant attention in the literature. Manlove~\cite{manlove1999stable} extended the results in~\cite{irving1994stable} to the general one-to-one setting (\textit{i.e.} incomplete bipartite graphs) and also showed that all strongly stable matchings have the same size and match the same set of vertices. Irving \etal~\cite{irving2003strong} further extended these results to the \HRT\ setting and gave an $O(m^2)$ algorithm for the strongly stable matching problem, which was later improved to $O(mn)$ by Kavitha~\etal~\cite{KavithaMMP07}. Manlove~\cite{Manlove02Lattice} studied the structure of the set of \SSM s and showed that, similar to the classical stable matchings, the set of \SSM s forms a distributive lattice. Kunysz~\etal~\cite{kunysz2016characterisation} showed that there exists a partial order with $O(m)$ elements representing all \SSM s and also provided an $O(mn)$ algorithm to construct such a representation. In the presence of edge weights, Kunysz~\cite{Kunysz18MaxWt} showed that when edge weights are small, the maximum weight \SSM\ problem can be solved in $O(mn)$ time, and in $O(mn\ \log(Wn))$ if the maximum weight of an edge is $W$. Strong stability with respect to restricted edges, namely forced, forbidden, and free edges, has been studied by Cseh and Heeger~\cite{CsehH20Restricted} and by Boehmer and Heeger~\cite{BoehmerH23ForcedForbidden}.

\vspace{0.1in}

\noindent\textbf{Organization of the paper:}

In Sections~\ref{sec:MinSumSSM},~\ref{sec:ForcedMatch}, and~\ref{sec:costVersion}, the objective of our problems is to minimize the total increase in quotas. In Section~\ref{sec:MINMAXProblem}, our objective is to minimize the maximum increase in quotas; there, we study the \MINMAXSS\ problem and its variant with bounded tie length. Finally, Section~\ref{sec:Conclusion} summarizes the contributions and discusses potential directions for future research. 
\section{MINSUM-SS problem}
\label{sec:MinSumSSM}

In this section, we present an efficient algorithm for the \hyperref[prob:minsumss]{\MINSUMSS} problem. Since the input instance does not admit a \SSM, we need to increase the quotas of certain hospitals to obtain $G'$. Our algorithm (pseudo-code given in Algorithm~\ref{algo:MinSum}) involves a sequence of proposals from hospitals and is inspired by the hospital-oriented algorithm for super-stability~\cite{swat/IrvingMS00}.

The algorithm begins with every resident unmatched, or equivalently, with every resident matched to the dummy hospital $\bot$. Call this matching $M'$ (see line~\ref{algoMinSum:InitMatch} of Algorithm~\ref{algo:MinSum}).
During the course of the algorithm, let $h$ be a hospital that is under-subscribed in $M'$, and $t$ be the most preferred rank in \prefh\ at which $h$ has not yet made a proposal. Then,  $h$ simultaneously proposes to all residents at rank $t$ in \prefh\ (see Line~\ref{algoMinSum:propose}). Since a hospital $h$ proposes to all the residents at a particular rank simultaneously, it may lead to the over-subscription of that hospital. 
A fully subscribed or over-subscribed hospital does not propose further, and the sequence of proposals terminates when either no hospital is under-subscribed or all under-subscribed hospitals have exhausted proposing to all residents on their preference lists. When a resident $r$ receives a proposal from $h$, the resident accepts or rejects the proposal based on the resident's preference between $h$ and its current matched partner $M'(r)$. 
Let $M'$ represent the set of matched edges when the proposal sequence terminates. 
Since $G$ does not admit a \SSM, at least one hospital $h$ must be over-subscribed in $M'$.
Let $G'$ denote the instance with the modified quotas where the quota of each hospital $h\in\HH$ is set to $q'(h)=\max\{q(h),|M'(h)|\}$. The algorithm returns the augmented instance $G'$ and the matching  $M'$.
Observe that each hospital proposes along each edge at most once; hence, the running time of the algorithm is $O(m)$.

\begin{algorithm}
    \caption{Algorithm for \MINSUMSS }\label{algo:MinSum}
    \DontPrintSemicolon
    \SetAlgoLined
    $M'$ = $\{ (r, \bot )  \ \  \ | \ \ \  \mbox{for every resident $r \in \RR$ }\}$ \label{algoMinSum:InitMatch}\;
    \While{$\exists$  $h$ that is under-subscribed in $M'$ w.r.t. $q(h)$ and $h$ has not exhausted $\prefh$}{
         $h$ proposes to all residents at the most preferred rank $t$  that $h$ has not yet proposed\label{algoMinSum:propose}\;
         \For {every resident $r$ that receives a proposal from $h$} {
              \If { $h \succ_r M'(r)$ } 
                  { $M' = M' \setminus \{ (r, M'(r))\} \cup \{(r, h)\} $ }
         }
    }
     $G'$ is the same as $G$, except quotas are set as follows \;
     For each $h\in \HH$, set $q'(h)=\max\{q(h),|M'(h)|\}$\label{algoMinSum:capUpdate}\;
    \Return $G'$ and $M'$ \label{algoMinSum:return}
\end{algorithm}

Next, we prove the correctness and optimality of our algorithm. 
\begin{lem}\label{lem:SSMInstance}
    The matching $M'$ returned by Algorithm~\ref{algo:MinSum} is a \SSM\ in the augmented instance $G'$.
\end{lem}
\begin{proof}
By the way quotas of the hospitals are set in $G'$, it is clear that $M'$ is a valid matching in $G'$. 
    Suppose for contradiction, $M'$ is not strongly stable in $G'$. 
    This implies that there exists a \SBP, say $(r,h)$, with respect to $M'$ in $G'$.  
    Therefore, $h\succ_r M'(r)$ and there exists $r'\in M'(h)$ such that $r\succeq_h r'$. Since hospitals propose in order of their preference list, $h$ must have proposed $r$ during the course of Algorithm~\ref{algo:MinSum}. Since $h\not= M'(r)$, the resident $r$ must have rejected the proposal from $h$. Thus, at the time when $r$ rejected $h$, the resident $r$ must have been matched to a better-preferred hospital, say $h'$. Since during the course of the algorithm residents improve their matched hospital, the final matched hospital $M'(r)$ of $r$ must be such that $M'(r)\succeq_r h' \succ_r h$. This contradicts the fact that $h \succ_r M'(r)$ and completes the proof. 
\qed\end{proof}

To prove the optimality of our capacity increase, we establish useful properties of {\em any} augmented instance $\tG$ (not necessarily optimal), obtained from $G$, such that $\tG$ admits a strongly stable matching. Let $\tM$ be a \SSM\ in $\tG$. 
In Claim~\ref{cl:mustBeMatched}, we show that if a resident $r$ is matched to $h$ (not equal to $\bot$) in $M'$ output by Algorithm~\ref{algo:MinSum}, then 
$r$ is matched in $\tG$ and is matched to either $h$ or a better-preferred hospital in $\tG$.

\begin{cl}\label{cl:mustBeMatched}
    Let $r$ be matched to $h \in \HH$ in $M'$ at the end of Algorithm~\ref{algo:MinSum} ($h\neq \bot$). Then the resident $r$ is matched in $\tM$, and $\tM(r)\succeq_r h$.
\end{cl}

\begin{proof}

We proceed by contradiction. Suppose the claim fails, \textit{i.e.}, there exists a resident $r$ matched to some hospital $h\in\HH$ in $M'$ such that $h\succ_{r}\widetilde{M}(r)$. Since $M'(r)=h\in \HH$ (as $h\neq \bot$), hospital $h$ proposed to $r$ during the execution of Algorithm~\ref{algo:MinSum}, and $r$ accepted the proposal. In particular, $r$ received at least one proposal. Among all proposals made during Algorithm~\ref{algo:MinSum} that violate the claim, consider the \textit{first} such proposal in the proposal sequence\textemdash that is, the first proposal by some hospital $h$ to some resident $r$ such that $h\succ_{r}\widetilde{M}(r)$. 

Since $h\succ_{r}\widetilde{M}(r)$, and $\widetilde{M}$ is strongly stable, $h$ must be fully subscribed in $\widetilde{M}$ with better preferred residents than $r$ with respect to its capacity $\tilde{q}(h)$. Since $\tilde{q}(h)\geq q(h)$, there must be at least $q(h)$ many neighbors for $h$ in $\widetilde{M}(h)$, which $h$ strictly prefers to $r$. In Algorithm~1, $h$ proposes to all the residents in $\widetilde{M}(h)$ before proposing $r$.

Since $h$ proposes to $r$ in Algorithm~1, at least one resident, say $r_1$, in $\widetilde{M}(h)$ must have rejected the proposal of $h$ before $h$ proposed to $r$. Thus $M'(r_1)\succ_{r_1} h=\widetilde{M}(r_1)$. This contradicts the assumption that $h$ to $r$ is the first proposal such that $h\succ_{r} \widetilde{M}(r)$.
\qed\end{proof}

In the next claim, we show that any hospital that remains under-subscribed in $M'$with respect to $q(h)$ continues to remain under-subscribed (to the same extent or more) in a strongly stable matching $\tM$ of any augmented instance $\tG$.
\begin{cl}\label{cl:UnderSubscribed}
    Let $M'$ be the output of Algorithm~\ref{algo:MinSum} and $h\in \HH$ be any hospital such that $|M'(h)|< q(h)$. Also, assume that $\widetilde{G}$ is an augmented instance obtained from $G$, and $\widetilde{M}$ is a \SSM\ in $\widetilde{G}.$ Then, $|\tM(h)|\le |M'(h)|$. 
\end{cl}
\begin{proof}
    Since $|M'(h)|< q(h)$, the hospital $h$ exhausted proposing all residents in $\prefh$ during the execution of Algorithm~\ref{algo:MinSum}. Clearly, all neighbors of $h$ received proposals from $h$. If there exists any resident, say $r$, who rejected $h$ during the execution of Algorithm~\ref{algo:MinSum}, then $r$ must have been matched
    in $M'$ to $M'(r)$ where $M'(r) \succ_r h$. 
    Using  Claim~\ref{cl:mustBeMatched}, we conclude that  $\tM(r) \succeq_r M'(r) \succ_r h$. 
    Thus, no resident who rejected $h$ during the execution of Algorithm~\ref{algo:MinSum} can be matched to $h$ in $\tM$, implying that $|\tM(h)| \le |M'(h)|$.  
\qed\end{proof}

Now, we show that the total increase in quotas of all hospitals incurred by Algorithm~\ref{algo:MinSum} is optimal.

\begin{lem}
\label{lem:optimalK}
The total quota increase by Algorithm~\ref{algo:MinSum} is optimal.
\end{lem}
\begin{proof}
Let $\RRM \subseteq\RR$ be the set of residents who received some proposal during the execution of Algorithm~\ref{algo:MinSum} and hence residents in $\RRM$ are matched in $M'$. By Claim~\ref{cl:mustBeMatched}, every $r\in\RR_m$ must be matched in $\tM$.
    Let $\HH_u$ be the set of  hospitals such that $|M'(h)|<q(h)$, and $\HH_f = \HH \setminus \HH_u $.
    Let the $\RRMU$ denote the set of residents matched in $M'$ to hospitals in $\HH_u$.
    By Claim~\ref{cl:UnderSubscribed}, the quota utilization over all hospitals $\HH_u$ in a strongly stable matching $\tM$ of any instance $\tG$  must be at most $|\RRMU|$.  This implies that at least $|\RRM \setminus \RRMU|$ many residents must be matched to hospitals in $\HH_f$ in the matching $\tM$. Let $k = |\RRM \setminus \RRMU| - \sum_{h\in \HH_f} q(h)$. Thus, the total quota increase in any instance $\tG$ is at least $k$. Algorithm~\ref{algo:MinSum} increases the quotas of hospitals in $\HH_f$ only and matches the residents in $\RRM \setminus \RRMU$ to hospitals in $\HH_f$. Thus, the total quota increase of hospitals in $G'$ is exactly $k$ which is optimal.
\qed\end{proof}

Lemma~\ref{lem:SSMInstance} and Lemma~\ref{lem:optimalK}  together imply Theorem~\ref{thm:MinSum}.

It is well known that when an \HRHT\ instance admits a strongly stable matching, all its strongly stable matchings match the same set of residents~\cite{irving2003strong}. In a similar spirit, we prove that all optimal solutions of a given \hyperref[prob:minsumss]{\MINSUMSS} instance match the same set of residents.

\begin{thm}\label{theo:MinSumProp}
Let $G'$ be the instance returned by Algorithm~\ref{algo:MinSum} and $\RRM$ denote the set of residents matched in the \SSM\ $M'$. Then, for any optimal augmented instance $G_{opt}$, the set of residents matched in any \SSM\ of $G_{opt}$ is exactly $\RRM$.
\end{thm}

\begin{proof}
    Theorem~\ref{thm:MinSum} asserts that the instance $G'$ returned by Algorithm~\ref{algo:MinSum} is an optimal augmented instance for $G$. Let $M_{opt}$ be a \SSM\ in $G_{opt}$. Applying Claim~\ref{cl:mustBeMatched}, we know that $M_{opt}$ must match all residents in $\RRM$.  If $M_{opt}$ matches any resident $r\notin \RRM$, then $M_{opt}$ must match more than $|\RRM|$ many residents for the instance $G_{opt}$. Using Claim~\ref{cl:UnderSubscribed}, we observe that any hospital $h$ that is under-subscribed in $M'$ with respect to $q(h)$ is matched to at most $|M'(h)|$ many residents in $M_{opt}$. Thus, the matching $M_{opt}$ must match $r$ to a hospital $h$ such that $|M'(h)|\ge q(h)$. Therefore, the total increase in quotas by $G_{opt}$ is more than that of $G'$. This contradicts the optimality of $G_{opt}$.
\qed\end{proof}

 Using Claim~\ref{cl:UnderSubscribed} and Theorem~\ref{theo:MinSumProp}, we have the following corollary
\begin{corollary}\label{cor:fullHospitals}
    Let $G'$ be the instance returned by Algorithm~\ref{algo:MinSum}. Also, assume that $G_{opt}$ is any optimal augmentation and $M_{opt}$ is a \SSM\ in $G_{opt}$. Then, $|M'(h)|\ge q(h)$ for a hospital $h$ implies that $|M_{opt}(h)|\ge q(h)$. Moreover, if $|M'(h)|< q(h)$, then $|M_{opt}(h)|=|M'(h)|$.
\end{corollary}

Now, let us consider a variant of the \hyperref[prob:minsumss]{\MINSUMSS} problem where our goal is to determine the existence of an augmented instance that admits a resident-perfect \SSM. Let us denote this problem by {\sf \textsc{MinSum-SS-RP}}.  Chen and Cs{\'{a}}ji~\cite{ChenCsaji23CapAAMAS} studied a special case of this problem called {\sf \textsc{MinSum Cap Stable and Perfect}} problem. Given an \HR\ instance (strict list), say $G$, and a budget $\ell$, the {\sf \textsc{MinSum Cap Stable and Perfect}} problem asks whether it is possible to obtain an augmented instance $G'$ from $G$, only by increasing the quotas of some hospitals, such that $G'$ admits a resident-perfect stable matching, and the sum of the increase in quotas over all hospitals is at most $\ell$.  They showed that this problem is NP-complete even for a very restricted case. Therefore, we conclude that   {\sf \textsc{MinSum-SS-RP}} problem is NP-complete.
\section{MINSUM-SS with Forced Edges} \label{sec:ForcedMatch}

In this section, we consider the \hyperref[prob:minsumssfe]{\MINSUMSSFE} problem. Here, in addition to an \HRHT\ instance $G$, we are given a subset $Q\subseteq E$. Our goal in this problem is to determine whether there exists an augmented instance $G'$ that admits a \SSM\ $M'$ with $Q\subseteq M'$. Note that if $Q$ contains two edges $(r,h_1)$ and $(r,h_2)$ for some resident $r$ and distinct hospitals $h_1\neq h_2$, then no such augmented instance $G'$ exists. Therefore, we assume that for any two edges $(r_1,h_1), (r_2,h_2)\in Q$, we have $r_1\neq r_2$.

We remark that such an augmented instance need not exist even when $|Q|=1$. For instance, consider an instance with one resident $r$ and two hospitals $h_1,h_2$ where $r$ prefers $h_1$ over $h_2$. Let $Q=\{(r,h_2)\}$. There is no way to augment the quotas of hospitals to get a \SSM\ which contains $Q$. We show that the \hyperref[prob:minsumssfe]{\MINSUMSSFE} problem admits a polynomial-time algorithm.  
Whenever an augmentation is possible, we output the optimally augmented instance.   

\subsection{Overall idea of our algorithm}

First, we present the overall idea of our algorithm.
Our algorithm constructs a pruned graph by deleting certain edges\textemdash specifically, any edge whose inclusion in a matching, together with the forced edges $Q$, would introduce a strong blocking pair.  The pruned graph may or may not admit a \SSM. Hence, we use Algorithm~\ref{algo:MinSum} (for the \hyperref[prob:minsumss]{\MINSUMSS} problem)  to obtain an augmented instance. While the augmented instance admits a strongly stable matching, the matching along with the forced edges $Q$ may not be strongly stable. To finally determine whether an augmentation is possible, we crucially use the analog of the rural hospitals theorem for the \hyperref[prob:minsumss]{\MINSUMSS} problem, established as  Theorem~\ref{theo:MinSumProp} and Corollary~\ref{cor:fullHospitals}.
  
 For a hospital $h\in \HH$, let $Q(h)=\{r\ | \ (r,h)\in Q\}$.
 If such an augmented instance exists, then the quota of each hospital $h$ must be at least $|Q(h)|$. Therefore, for each hospital $h$, without loss of generality, we assume $q(h)\ge|Q(h)|$.

\subsection{Pruned instance}
Let $(r,h)\in Q$. Suppose there exists a resident $r'$ such that $r'\succeq_{h} r$. Then $r'$ cannot be matched to any $h'$ where $h\succ_{r'} h'$ in a \SSM\ $M$ containing the edge $(r,h)$; otherwise,  $(r',h)$ will block $M$. We call such a resident $r'$ a \textit{distracting resident} for the forced edge $(r,h)\in Q$. Similarly, for the edge $(r,h)\in Q$, if there exists a hospital $h'$ such that $h'\succ_{r} h$, then $h'$ cannot be matched to any resident $r'$ where $r\succeq_{h'} r'$ in any \SSM\ which contains the edge $(r,h)$. We call such a hospital $h'$  a \textit{distracting hospital} for the forced edge $(r,h)$. 
The notion of distracting agents is borrowed from~\cite{GokhaleSNV24AAMASCap}, where a special case of our problem is considered in the strict list setting with $|Q| = 1$. 
For each distracting resident/hospital, we remove any edge along which the distracting agent cannot be matched in any \SSM\ containing the edge $(r,h)$. It is possible that the set of deleted edges may contain edges from $Q$. If it happens, then there exists an edge that needs to be deleted, but is also a forced edge. Hence, there does not exist an augmented instance $G'$ which admits a \SSM\ containing the set $Q$, and therefore, we output ``No augmentation possible to get a \SSM\ containing Q".  In the remaining part of the section, we assume that the set of deleted edges does not contain any edge from $Q$.

 Given the edge set $Q$ of edges, we define a set $\RR_d$ of distracting residents and a set $\HH_d$ of distracting hospitals corresponding to all the edges in $Q$. Thus,  
$ \RR_d=\{r'\ | \ r'\succeq_{h} r \text{ for some } (r,h)\in Q\}$ and $\HH_d=\{h'\ | \ h'\succ_{r} h \text{ for some } (r,h)\in Q\}$. Next, we formally describe the pruning steps discussed above. We use the edge set $E_d$ to keep track of deleted edges.

\begin{enumerate}[label=\arabic*.]
    \item \label{itm:Pruned1} Initialize $E_d=\emptyset$. For each edge $(r,h)\in Q$: 
    \begin{enumerate}[label = (\alph*)]
    \item\label{itm:pruned11} For each $r'\in \RR_d$, if $r'$ is distracting for edge $(r,h)$, then we delete all $h'$ from $\prefrp$ and $r'$ from $\prefhp$ such that $h\succ_{r'} h'$. Add all such $(r',h')$ to $E_d$.

    \item\label{itm:pruned12} For each $h'\in \HH_d$, if $h'$ is distracting for edge $(r,h)$, then we delete all $r'$ from $\prefhp$ and $h'$ from $\prefrp$ such that $r\succeq_{h'} r'$. Add all such $(r',h')$ to $E_d$.
    
\end{enumerate}

\item  Let $Q_{\RR}=\{r\ |\ (r,h)\in Q\}$. 
Since every resident in $Q_{\RR}$ has a forced match, we delete all edges incident on $r \in Q_{\RR}$. We include these deleted edges in $E_d$. Observe that after this step, the edges in $Q$ are in $E_d$. To account for matching residents in $Q_{\RR} $ to their forced partners, we appropriately reduce the quotas of the corresponding hospitals. 

\end{enumerate}

 We call this pruned instance $G_p$. We note that during this pruning process, some residents may become isolated, and the quotas of some hospitals may go down to zero. We denote the set of edges in $G_p$ by $E_p$. That is, $E_p=E\setminus E_d$. Thus, the pruned instance $G_p=(\RR\cup \HH, E_p)$. For each hospital $h$, its quota in $G_p$ is denoted by  $q_p(h)$ and  $q_p(h)=q(h)-|Q(h)|$. It is easy to observe that both steps described above run in $O(m)$ time. Therefore, the pruned graph $G_p$ can be constructed in $O(m)$ time.

\subsection{Our algorithm}
 We now describe our algorithm. Observe that the pruned instance $G_p$ may or may not admit a \SSM. If $G_p$ does not admit a \SSM, we run Algorithm~\ref{algo:MinSum} for the \hyperref[prob:minsumss]{\MINSUMSS} problem to obtain an augmented instance $G_p'$ that admits a \SSM, say $M_p'$. For the instance $G_p'  = (\RR \cup \HH, E_p)$, let $q'_p(h)$ denote the quota of hospital $h\in\HH$. Even in the case when the pruned graph  $G_p$ itself admits a \SSM, for convenience of notation we rename the instance $G_p$ as $G_p'$ and a \SSM\ in $G_p'$ as $M_p'$. Although $M_p'$ is a \SSM\ in $G_p'$, for the matching $M_p'\cup Q$ to be a \SSM, we need the following checks. 

\begin{enumerate}[left=0pt, label = Step \arabic*.]
        \item \label{itm:step1}  If there exists some  $h \in \HH_d$ such that $|M_p'(h)| < q(h)$ (\textit{i.e.} $h$ is under-subscribed with respect to its original quota), then we declare ``No augmentation possible to get a \SSM\ containing $Q$''. This step is justified in Lemma~\ref{lem:noSolMinSumMatch}.
        
        \item \label{itm:step2} Otherwise, we obtain a further augmented instance $G_p^*$ and a matching $M_p^*$ as described next. For a resident $r\in\RR$, let $last(r)$ denote the least-preferred hospital of $r$ in the pruned instance $G_p$. Suppose $r'\in \RR_d$ is a distracting resident for some edge $(r,h)$. This implies $r'\succeq_h r$. We observe that for any augmented instance of $G$, no \SSM\ containing the edge $(r,h)$ can match the resident $r'$ to a hospital $h'$ such that $last(r') \succ_{r'} h'$, otherwise $(r', last(r'))$ blocks the resulting matching. For the same reason as above, no $r'\in \RR_d$ can be left unmatched.  Let $\RR_d^u$ denote the residents in $\RR_d$ that are left unmatched in the matching $M_p'$.  Then, for every $r' \in \RR_d^u$ in the instance $G_p'$, we match $r'$ to the hospital $last(r')$ to get the matching $M_p^*$ and increase the quota of $last(r')$ by 1.  We denote by $q^*_p(h)$ the quota of every hospital in the instance $G_p^*$.

        \item \label{itm:step3} For each $(r, h) \in Q$ we add $(r,h)$ to $M_p^*$ and increase the quota of $h$ by 1. We also restore all the edges in $E_d$.  With a slight abuse of notation, we call this instance $G_p^*$ with quota $q_p^*(h)$ for all $h\in\HH$. We return $G_p^*$ and the matching $M_p^*$.
\end{enumerate}

This completes the description of our algorithm. 

\vspace{0.1in}

\noindent\textbf{Running time of our algorithm.} Recall that the pruned graph $G_p$ can be constructed in $O(m)$ time. As noted in Section~\ref{sec:MinSumSSM}, executing Algorithm~\ref{algo:MinSum} takes $O(m)$ time to compute $G_p'$ and $M_p'$. Constructing $M_p^*$ from $M_p'$ also takes $O(m)$ time. Thus, the overall running time of our algorithm is $O(m)$. 

\subsection{Correctness and optimality of our algorithm}

We first prove the strong stability of the matching $M_p^*$ output by our algorithm.

\begin{cl}\label{cl:GpEquiG1}
    No edge in $E_d$ blocks the matching $M_p^*$ output at the end of ~\ref{itm:step3}
\end{cl}
\begin{proof}    
      We first note that no edge in $Q$ can be strongly blocking for $M_p^*$ because $Q\subseteq M_p^*$. Therefore, we only need to consider the edges in $E_d\setminus Q$.  Observe that for each edge $(r',h')\in E_d\setminus Q$, either $r'\in Q_{\RR}$ or $r'\notin Q_{\RR}$. We consider both of these cases below. Note also that each $r'\in Q_{\RR}$ is matched in $M_p^*$.

      \begin{enumerate}[left=0pt, label = Case \arabic*.]
     \item $(r',h')\in E_d\setminus Q$ and $r'\in Q_{\RR}$: Suppose $M_p^*(r')=\Hat{h}$. It is easy to observe that if $\Hat{h} \succ_{r'} h'$, then $(r',h')$ cannot be a blocking pair with respect to $M_p^*$. Therefore, let us assume that $h' \succ_{r'} \Hat{h}$. Then $h'$ is a distracting hospital for the edge $(r',\Hat{h})$. By construction, the pruned instance $G_p$ does not contain any edge $(r'',h')$ such that $r' \succeq_{h'} r''$. This implies $h'$ must be fully subscribed in $M_p^*$ with respect to its augmented quota, $q_p^*(h') \ge q_p(h')$, and all residents matched to $h'$ are strictly better than $r'$. Thus, $(r',h')$ cannot be a blocking pair with respect to $M_p^*$.

     \item $(r',h')\in E_d\setminus Q$ and $r'\notin Q_{\RR}$: In this case, at least one of the following holds: (i) $r'$ is a distracting resident, or (ii) $h'$ is a distracting hospital. If $r'$ is a distracting resident, say for an edge $(r,h)\in Q$, then $last(r')\succ_{r'} h'$. This implies $(r',h')$ cannot block $M_p^*$ because $M_p^*(r')\succ_{r'} h'$. Therefore, let us assume that $h'$ is a distracting hospital. Suppose $(r,h)\in Q$ is one of the edges for which $h'$ is distracting. This implies $h' \succ_r h$. Since $(r,h)\in M_p^*$ and $M_p^*$ is a \SSM\ in $G_p^*$, the hospital $h'$ must be fully subscribed in $M_p^*$ with respect to its augmented quota, $q_p^*(h') \ge q_p(h')$, and all residents matched to $h'$ are strictly better than $r'$. Thus, $(r',h')$ cannot be a blocking pair with respect to $M_p^*$. 
     \end{enumerate}
     Therefore, no edge in $E_d$ blocks the matching $M_p^*$ output at the end of ~\ref{itm:step3}\qed
     \end{proof}

Having proved that the matching $M_p^*$ is strongly stable in the augmented instance $G_p^*$ of $G$, we proceed to prove the optimality of $G_p^*$. 

We first show, using Lemma~\ref{cl:GpEquiG2}, that it suffices to start with the pruned instance $G_p$.

\begin{lem}\label{cl:GpEquiG2}
    Let $\tG$ be any augmented instance obtained from $G$ such that $\tG$ admits a \SSM\ $\tM$ where $Q\subseteq \tM$. Then, the following hold:
    \begin{itemize}[left=20pt]
    \item[(i)] $\tM\cap E_d=Q$. \vspace{0.05in}
    \item[(ii)] Moreover, $\tM\setminus Q$ is a \SSM\ in some augmented instance of $G_p$.
    \end{itemize}
\end{lem}

\begin{proof}
   Recall that $Q\subseteq E_d$. Therefore, we need to show that if there exists an edge $(r',h')\in E_d\setminus Q$, then $(r',h')\notin \tM$.  Recall the resident set $Q_{\RR}$. The fact that $Q\subseteq \tM$ implies that if an edge $(r',h')\in E_d\setminus Q$ such that $r'\in Q_{\RR}$, then $(r',h')\notin \tM$.
   
   Now, let us assume that $(r',h')\in E_d\setminus Q$ is an edge such that $r'\notin Q_{\RR}$. For any such edge $(r',h')$, at least one of the following holds: (i) $r'$ is a distracting resident, or (ii) $h'$ is a distracting hospital. Let us first assume that $r'$ is a distracting resident for some edge $(r,h)\in Q$. This implies $r' \succeq_{h} r$. Also, the fact that $(r',h')\in E_d\setminus Q$ implies that $h\succ_{r'} h'$. Thus, if $(r',h')\in \tM$, then $(r',h)$ blocks the matching $\tM$ because $(r,h)\in \tM$; a contradiction. Now, let us assume that $h'$ is a distracting hospital for some edge $(r,h)\in Q$. This implies that $h' \succ_r h$. Also, since $(r',h')\in E_d$, we have $r\succeq_{h'} r'$.  Since $(r,h)\in \tM$, if $(r',h')\in \tM$, the edge $(r,h')$ blocks the matching $\tM$; a contradiction. This completes the first part of the proof. 
   
    First, we observe that $M'=\tM\setminus Q$ is indeed a matching in some augmentation of~$G_p$. 
    This is because if a hospital $h$ is such that $|M'(h)|> q_p(h)$, we set the quota of $h$ equal to $|M'(h)|$ to get an augmented instance of $G_p$. Let this augmented instance of $G_p$ be $\tG_p$ where the quota of each hospital $h$ is $\Tilde{q}_p(h)$. 
    

   Now, we show that $M' = \tM \setminus Q$ is a \SSM\ in $\tG_p$. Note that $\tG_p$ does not contain any of the edges in $E_d$. 
   Suppose, for contradiction, that $M'$ is not a \SSM\ in $\tG_p$. Then there exists a blocking edge, say $(r, h)$, with respect to $M'$ in $\tG_p$. Since $(r, h)$ is a blocking pair with respect to $M'$, we must have $h \succ_r M'(r)$ and either (i) $r \succ_h r'$ for some $r' \in M'(h)$ or (ii)  $|M'(h)| < \Tilde{q}_p(h)$. If $h \succ_r M'(r)$, then because $E_p \subseteq E$, it follows that $h \succ_r \tM(r)$. Now, if $r \succ_h r'$ for some $r' \in M'(h)$, then, as discussed above, $r \succ_h r'$ for the same $r' \in \tM(h)$. On the other hand, if $|M'(h)| < \Tilde{q}_p(h)$, then, because $|\tM(h)| - |M'(h)| = |Q(h)|$, we have $|\tM(h)| < \Tilde{q}_p(h) + |Q(h)|$. This implies $|\tM(h)| < \Tilde{q}(h)$, where $\Tilde{q}(h)$ is the quota of hospital $h$ in the augmented instance $\tG$. Therefore, $(r, h)$ is a blocking pair with respect to $\tM$ in $\tG$, contradicting the assumption that $\tM$ is a \SSM\ in $\tG$. \qed
\end{proof}

Next, we justify \ref{itm:step1} of our algorithm.

\begin{lem}\label{lem:noSolMinSumMatch}
    If there exists a hospital $h'\in\HH_d$ such that $|M_p'(h')|<q(h')$, then the given instance of the \hyperref[prob:minsumssfe]{\MINSUMSSFE} problem has no solution.
\end{lem}
\begin{proof}
Applying Claim~\ref{cl:GpEquiG1} and Lemma~\ref{cl:GpEquiG2}, it suffices to start with the pruned graph $G_p$. Let $\tG_p$  be any (possibly) augmented instance obtained from $G_p$ such that $\tG_p$ admits a \SSM\ say $\tM_p$ and $Q\subseteq \tM_p$.  Assume that $h'$ is a distracting hospital for an edge $(r,h)\in Q$. This implies that $h' \succ_{r} h$. Let us first assume that $G_p$ admits a \SSM, say $M_p'$. The fact that for each hospital, all \SSM s fill its quota to the same extent for any instance~\cite{swat/IrvingMS00} and $|M_p'(h')|<q(h')$ implies that $h'$ is under-subscribed in each \SSM\ of $G_p$. Since $(r,h)\in \tM_p$ and $h' \succ_{r} h$, we conclude that $(r,h')$ blocks $\tM_p$. Therefore, let us assume that $G_p$ does not admit a \SSM. In this case, we execute Algorithm~\ref{algo:MinSum} on $G_p$ to get $G_p'$ and $M_p'$.  We note that for any hospital $\hat{h} \in \HH_d$ that remains under-subscribed in $M'_p$ with respect to its quota $q(\hat{h})$, Algorithm~\ref{algo:MinSum} did not augment its quota.   By Claim~\ref{cl:UnderSubscribed}, we know that $|\tM_p(h')| < q(h')$. Since $h' \succ_{r} h$ and $(r,h)\in \tM_p$, we conclude that $(r, h')$ blocks $\tM_p$. Since this is true for {\em any} augmented instance obtained from $G_p$, we conclude that no augmentation is possible to get a \SSM\ containing $Q$.\qed
 \end{proof}

\begin{lem}\label{lem:ForcedMatch}
    If our algorithm returns an augmented instance $G_p^*$, then $G_p^*$ is optimal.
\end{lem}

\begin{proof}
        As mentioned above, it suffices to start with the pruned graph $G_p$. We also note that if the given instance of the \hyperref[prob:minsumssfe]{\MINSUMSSFE} problem admits a solution, then each $r\in\RR_d^u$ (unmatched distracting resident) must be matched to one of its neighbors $h$ such that $h\succeq_r last(r)$. 

            
            By Theorem~\ref{thm:MinSum}, we know that $G_p'$ is an optimal augmented instance obtained from $G_p$, such that $G_p'$ admits a \SSM. If $\RR_d^u = \emptyset$, then there is nothing to prove as optimality directly follows by Theorem~\ref{thm:MinSum}. Therefore, assume that $\RR_d^u \neq \emptyset$ and hence, some resident $r\in\mathcal{R}_d^u$ remains unmatched in $M_p'$. This implies that every neighbor $h$ of each $r\in\mathcal{R}_d^u$ in $G_p'$ (that is, $h$ such that $h\succeq_r last(r)$) is fully subscribed with respect to  $q_p'(h) \ge q_p(h)$. By Corollary~\ref{cor:fullHospitals}, for $r\in\mathcal{R}_d^u$, for a neighbor $h$ such that $h\succeq_r last(r)$, we have $|\widetilde{M}_p(h)| \ge q_p(h)$. Thus, $\widetilde{G}_p$ and $\widetilde{M}_p$  must use at least $|\mathcal{R}_d^u|$ many extra capacities over those in $G_p'$. 
       For each $r\in \RR_d^u$, we match $r$ to the corresponding least-preferred hospital $last(r)$ and increase the capacity of $last(r)$ by one. Therefore, the total capacity increase in $G_p^*$, across all hospitals, is at most $|\RR_d^u|$ over those in $G_p'$. Thus, we conclude that $M_p^*$ and $G_p^*$ are optimal solutions for the given instance. 
        This completes the proof of this lemma.\qed
 \end{proof}

 Using Lemma~\ref{lem:noSolMinSumMatch} and Lemma~\ref{lem:ForcedMatch} we conclude Theorem~\ref{thm:ForcedEdges}.

\section{MINSUM-COST Problem}\label{sec:costVersion}

In this section, we consider the \hyperref[prob:minsumcost]{\MINSUMCOST} problem and show that this problem is NP-hard even when the costs of the hospitals are $0$ or $1$. We prove this hardness result by a reduction from \MSAT. 
The \MSAT\ is a variant of the Boolean satisﬁability problem in which the input is a conjunction of clauses. Each clause is a disjunction of exactly three variables, and no variable appears in negated form. The goal is to determine whether there exists a truth assignment in which exactly one variable per clause is set to true. This problem is known to be NP-complete~\cite{schaefer1978complexity,garey1979NPcomplete}, even when each variable occurs in at most three clauses~\cite{denman2009using}. 

\subsection{Reduction} 

Let $\II$ be an instance of the \MSAT, where each variable occurs in at most three clauses.  Let $\{X_1,X_2,\ldots,X_{\beta}\}$ be the set of variables and $\{C_1,C_2,\ldots,C_{\alpha}\}$ be the set of clauses in $\II$, for non-negative integers $\alpha$ and $\beta$.

Given $\II$, we construct an instance $G=(\RR\cup\HH,E)$ of the \hyperref[prob:minsumcost]{\MINSUMCOST} problem with quota $q(h)=1$ for each $h\in\HH$ such that $G$ does not admit a \SSM. We also associate a cost $c(h) \in \{0,1\}$ with each $h\in \HH$. We show that there exists an augmented instance $G'=(\RR\cup\HH,E)$ with total augmentation cost zero that admits a \SSM\ if and only if there exists a truth assignment for $\II$ in which exactly one variable per clause is set to true.

For every variable $X_p$ in the instance $\II$, the reduced instance $G$ contains one resident $a_p$ and one global hospital $v_p$. The resident $a_p$ is connected only to the hospital $v_p$ in $G$. The hospital $v_p$ is connected to the resident $a_p$ as well as to additional residents, which we describe below. Corresponding to each clause in $\II$, the reduced instance $G$ contains a clause gadget.

\begin{figure}[t]
    \begin{subfigure}{.4\textwidth}
         \centering
            \begin{tabular}{ll}
			$b_i^s$ :& $v_i, w^s$\\[2pt]
                $b_j^s$ :& $v_j, w^s$\\[2pt]
                $b_k^s$ :& $v_k, w^s$\\[2pt]
			$d_1^s$ :& $w^s$\\[2pt]
                $d_2^s$ :& $w^s$\\[2pt]
		\end{tabular}
            \caption{}
        \label{subfig:prefA}
    \end{subfigure}%
    \begin{subfigure}{.4\textwidth}
        \centering
            \begin{tabular}{ll}
                & \\[2pt]  
                & \\[2pt] 
                $w^s$ :&  $(b_i^s, b_j^s, b_k^s), (d_1^s,d_2^s)$\\[2pt]  
                & \\[2pt]  
                & \\[2pt]  
            \end{tabular}
            \caption{}
            \label{subfig:prefB}
\end{subfigure}
\caption[Gadget used in the hardness reduction for the \MINSUMCOST\ problem]{\label{fig:preflistsMinCost}%
Gadget $G_s$ used in the hardness reduction for the \hyperref[prob:minsumcost]{\MINSUMCOST} problem (i) Preference lists of residents in gadget \gt{s}. (ii) Preference list of the hospital in gadget \gt{s}.
}
\end{figure}

\vspace{0.1in}

\noindent{\bf Clause gadget:} Let \ct{s} = $(X_i \vee X_j \vee X_k)$ be a clause in the instance $\II$. Corresponding to this clause, there is a gadget \gt{s} in the reduced instance $G$. The gadget \gt{s} consists of a resident set $\RR_s = \{b_i^s, b_j^s, b_k^s, d_1^s, d_2^s\}$ and a hospital set $\HH_s = \{w^s\}$. The preference lists of residents and hospitals in \gt{s} are given in Fig.~\ref{fig:preflistsMinCost}.
For $p\in\{i,j,k\}$, the preference list of a resident $b_p^s$ corresponding to the variable $X_p$  consists of two hospitals -- the global hospital $v_p$ and the hospital $w^s$ which is internal to the gadget \gt{s}. Assume that in $\II$, the variable $X_i$ appears in three clauses, namely \ct{s}, $C_{i_1}$, and $C_{i_2}$; the variable $X_j$ appears in \ct{s}, $C_{j_1}$, and $C_{j_2}$; and the variable $X_k$ appears in \ct{s}, $C_{k_1}$, and $C_{k_2}$. Then the preference lists of residents $a_i,a_j,a_k$ and global hospitals $v_i, v_j, v_k$ are as shown in Fig.~\ref{fig:preflists2MinCost}. In Fig.~\ref{fig:preflists2MinCost}, we assume that the variable $X_i$ appears in three clauses. If $X_i$ appears in only one clause $C_s$, then $v_i$ ranks $a_i$ at rank 1, and $b_i^s$ at rank 2. That is, the preference list of $v_i$ is {\sf Pref($v_i$)}= $a_i, b_i^s$. If $X_i$ appears in exactly two clauses, say $C_s$ and $C_{i_1}$, then $v_i$ ranks $a_i$ at rank 1 and $b_i^s, b_i^{i_1}$ at rank 2. That is, the preference list of $v_i$ is {\sf Pref($v_i$)}= $a_i, (b_i^s, b_i^{i_1})$.

\begin{figure}
\begin{subfigure}{.4\textwidth}
         \centering
    \begin{minipage}{\linewidth}
    \begin{eqnarray*}
    a_i &:& v_i\\
    a_j &:& v_j\\
    a_k &:& v_k\\
    \end{eqnarray*}
    \end{minipage}
    \caption{}
            \label{subfig:prefA2}
\end{subfigure}
    \begin{subfigure}{.4\textwidth}
        \centering
    \begin{minipage}{\linewidth}
    \begin{eqnarray*}
    v_i&:&\  a_i, (b_i^s, b_i^{i_1}, b_i^{i_2}) \\
    v_j&:&\  a_j, (b_j^s, b_j^{j_1}, b_j^{j_2}) \\
    v_k&:&\  a_k, (b_k^s, b_k^{k_1}, b_k^{k_2}) \\
    \end{eqnarray*}
    \end{minipage}
    \caption{}
            \label{subfig:prefB2}
\end{subfigure}
\caption[Preferences of global vertices in the hardness reduction for the \MINSUMCOST\ problem ]{\label{fig:preflists2MinCost}%
Preferences of global vertices in the hardness reduction for the \hyperref[prob:minsumcost]{\MINSUMCOST} problem. We assume that the variables $X_i,X_j$ and $X_k$ appear in three clauses. (i) Preference lists of the residents $a_i,a_j$ and $a_k$. (ii)~Preference list of global hospitals $v_i, v_j$ and  $v_k$.
}
\end{figure}

\vspace{0.1in}

\noindent{\bf Augmentation cost:} 
The augmentation cost for each hospital in the reduced instance $G$ is defined as follows: 

\begin{itemize}
    \item[1.] $c(v_p) = 0$ for all $p \in \{1, 2, \ldots, \beta\}$
    \item[2.] $c(w^s) = 1$ for all $s \in \{1, 2, \ldots, \alpha\}$
\end{itemize}

 This completes the description of our reduction.

\subsection{Correctness}
We claim that the reduced instance $G$ does not admit a \SSM. Recall that the quota of each hospital is one. Any \SSM\ $M$ in the reduced instance $G$ must match each resident $a_p$ to the corresponding global hospital $v_p$, as otherwise, $(a_p,v_p)$ strongly blocks $M$. For any $s \in \{1, 2, \ldots, \alpha\}$, the matching $M$ cannot leave $w^s$ unmatched; otherwise, some pair $(d_t^s, w^s)$ for $t \in \{1, 2\}$ would form a \SBP. Since $w^s$ has a quota of one, it cannot accommodate all three $b$-residents in the gadget \gt{s}. This implies that there exists a $b$-resident, say $b_j^s$, who is not matched to $w^s$ in $M$. As each global hospital has a quota of one and is already matched with its top-ranked $a$-resident, the resident $b_j^s$ remains unmatched in $M$. Therefore, the pair $(b_j^s, w^s)$ strongly blocks $M$. Hence, the reduced instance $G$ does not admit a \SSM.

Recall that $\II$ is an instance of \MSAT. A satisfying assignment for an instance of \MSAT\ is an assignment of variables such that for each clause, exactly one variable is set to true.
 
\begin{lem}\label{lem:MinCostRed2}
    If $\II$ admits a satisfying assignment, then there exists an instance $G'$ obtained from $G$ with an augmentation cost of zero such that $G'$ admits a \SSM.
\end{lem}
\begin{proof}
    Given a valid satisfying assignment for the instance $\II$ of the \MSAT, we construct the augmented instance $G'$ and a matching $M'$ as follows. For each variable $X_p$, if $X_p$ is set to true, we set the quota $q(v_p) = 1$; otherwise, we augment the quota of $v_p$ by $\delta(v_p)$, where $\delta(v_p)$ denotes the number of times the variable $X_p$ occurs in $\II$. That is, we set $q(v_p) = 1 + \delta(v_p)$, where $1$ is the original quota. Additionally, we set $q(w^s) = 1$ for all $s \in \{1, 2, \ldots, \alpha\}$. Since the augmentation cost $c(v_p) = 0$ for all $p \in \{1, 2, \ldots, \beta\}$, the total augmentation cost is zero.
    
    Next, we show that the augmented instance $G'$ admits a strongly stable matching. Let $M' = M_1 \cup M_2 \cup M_3$, where $M_1$, $M_2$, and $M_3$ are defined as follows: 

      \begin{eqnarray*} 
            M_1&=&\ \bigcup_{p=1}^{\beta} \{ (a_p,v_p)\}\\
            M_2&=&\ \bigcup_{v_p\ :\  q(v_p)>1} \{ (b_p^s,v_p)\ |\  (b_p^s,v_p)\in E \text{ for some } s\}\\
            M_3&=&\ \bigcup_{v_p\ :\  q(v_p)=1} \{ (b_p^s,w^s)\ | \  (b_p^s,w^s)\in E \text{ for some } s\}
        \end{eqnarray*}
        
    We now verify that the matching $M' = M_1 \cup M_2 \cup M_3$ is a valid matching. 
    \begin{itemize}
        \item We note that each resident is matched to at most one hospital. It is clear that each $a$-resident, say $a_p$, is matched to the corresponding global hospital $v_p$ (as defined by $M_1$). We claim that for each resident $b_p^s$, we have $|M'(b_p^s)| = 1$. This follows from the following argument. For any clause $C_s = (X_i \vee X_j \vee X_k)$, exactly one variable is set to true. Consequently, exactly one of the three global hospitals $v_i, v_j, v_k$ has its quota set to one, while the remaining two have their quotas set to their degrees in $G'$. Therefore, if $q(v_p) = 1$, then $M'(b_p^s) = \{w^s\}$ (as defined by $M_3$), and hence $|M'(b_p^s)| = 1$. On the other hand, if $q(v_p) > 1$, then $M'(b_p^s) = \{v_p\}$ (as defined by $M_2$), so again $|M'(b_p^s)| = 1$. Thus, every $a$-resident and $b$-resident is matched to exactly one hospital. It is trivial to observe that no $d$-resident is matched in $M'$.
        \item We also note that the matching $M'$ respects the quotas of all hospitals in $G'$. Specifically, (i) if a global hospital $v_p$ has $q(v_p) = 1$, then $M'(v_p)=\{a_p\}$ (as defined by $M_1$ and $M_3$); and (ii) if $v_p$ has $q(v_p) > 1$, then $|M'(v_p)|\le q(v_p)$, since in this case $q(v_p)$ equals the degree of $v_p$ in $G'$. Furthermore, as defined by $M_3$, each $w$-hospital is matched to exactly one resident.
    \end{itemize}

  To establish that $M'$ is strongly stable, we prove that no resident in $G'$ is involved in a strong blocking pair with respect to $M'$. Each $a_p$ is matched to its rank-1 hospital, and thus no $a$-resident participates in a strong blocking pair.
    
   We now argue that no $b$-resident participates in a strong blocking pair. Let $C_s = (X_i \vee X_j \vee X_k)$ be an arbitrary clause in $\II$. Without loss of generality, suppose that in the satisfying assignment, the variable $X_k$ is set to true. Then $q(v_k) = 1$, and by construction of $M'$, we have $M'(b_k^s) = w^s$. Since for the resident $b_k^s$, the rank-1 hospital $v_k$ is fully subscribed with a more preferred resident $a_k$, we conclude that $b_k^s$ does not participate in any strong blocking pair. The other two $b$-residents, $b_i^s$ and $b_j^s$, are matched to their rank-1 hospitals ($v_i$ and $v_j$, respectively), and thus they also do not participate in any strong blocking pairs. The $d$-residents in the gadget \gt{s} cannot block $M'$ since the only hospital they are connected to, $w^s$, is already fully subscribed with a better-preferred resident. Therefore, the matching $M'$ is a strongly stable matching.\qed
\end{proof}

To prove the other direction, we require the following two claims.

\begin{cl}{\label{cl:wsOne}} 
Let $G'$ be an instance obtained from $G$ with zero total augmentation cost. Assume that $G'$ admits a strongly stable matching $M'$. Then, for any gadget $G_s$, we have 
$|M'(w^s)|=1$ and $M'(w^s) \in \{b_i^{s}, b_j^{s}, b_k^{s}\}$.

\end{cl}
\begin{proof}
    Since the augmentation cost of $w^s$ is one, and the instance $G'$ is obtained with zero total augmentation cost, it follows that the quota of $w^s$ is not augmented in $G'$. Given that the original quota of $w^s$ is one, we have $|M'(w^s)| \le 1$. Moreover, $M'$ cannot leave $w^s$ unmatched, nor can it match $w^s$ to one of the $d$-vertices --- doing so would leave some $d$-vertex unmatched, which, together with $w^s$, would form a strong blocking pair with respect to $M'$. Therefore, $|M'(w^s)|=1$ and $M'(w^s) \in \{b_i^{s}, b_j^{s}, b_k^{s}\}$.\qed
\end{proof}

\begin{cl}{\label{cl:bpConsistency}}
    Let $G'$ be an instance obtained from $G$ with zero total augmentation cost. Assume that $G'$ admits a strongly stable matching $M'$. For any gadget \gt{s}, if  $(b_i^s,w^s)\in M'$, then $(b_i^{i_1},w^{i_1})\in M'$ and $(b_i^{i_2},w^{i_2})\in M'$.
\end{cl}

\begin{proof}
    Assuming that $(b_i^s,w^s)\in M'$, we show that $(b_i^{i_1},w^{i_1})\in M'$.  Suppose for contradiction that $M'(b_i^{i_1})\neq w^{i_1}$. We note that $b_i^{i_1}$ cannot remain unmatched in $M'$; otherwise, $(b_i^{i_1}, w^{i_1})$ is a \SBP\ with respect to $M'$. This implies that $M'(b_i^{i_1})=v_i$. Recall that the hospital $v_i$ is indifferent between residents $b_i^s$ and $M'(v_i)=b_i^{i_1}$. The fact that $v_i \succ_{b_i^s} w^s$ and $M'(b_i^s)=w^s$ implies that the edge $(b_i^s,v_i)$ is a \SBP\ with respect to $M'$. The proof for $(b_i^{i_2},w^{i_2})\in M'$ is analogous.\qed
\end{proof}

\begin{lem}\label{lem:MinCostRed1}
    If there exists an instance $G'$ obtained from $G$ with zero total augmentation cost such that $G'$ admits a \SSM, say $M'$, then
    the instance $\II$ admits a satisfying assignment.
\end{lem}
\begin{proof}
    By Claim~\ref{cl:wsOne}, we know that $|M'(w^s)|=1$ and $M'(w^s) \in \{b_i^{s}, b_j^{s}, b_k^{s}\}$ for each $s\in\{1,2,\ldots,\alpha\}$. We obtain the assignment of variables as follows: for the gadget \gt{s},  if $M'(b_k^s)=w^s$, then we set $X_i$ and $X_j$ to false and $X_k$ to true in clause \ct{s}. Clearly, clause \ct{s} is satisfied, and exactly one variable is set to true in \ct{s}. Using Claim~\ref{cl:bpConsistency}, we note that this assignment is consistent. Thus, every clause has exactly one variable in that clause set to true, and we have a satisfying assignment for $\II$.\qed
\end{proof}

 This completes the hardness reduction, showing that the decision version of the \hyperref[prob:minsumcost]{\MINSUMCOST} problem is {\sf NP}-hard for the budget $\ell =0$. The above reduction also shows that no approximation algorithm exists for the \hyperref[prob:minsumcost]{\MINSUMCOST} problem; that is, it cannot be approximated within any multiplicative factor.  



Now, we show that for our reduced instance $G$, resident preferences, as well as hospital preferences, are derived from master lists. 

\begin{cl}
    In our reduced instance $G$, residents' preference lists are derived from a master list. Moreover, hospitals' preference lists are also derived from a master list.
\end{cl}

\begin{proof}
    To show that residents' preference lists are derived from a master list, we give a master list of hospitals as $$\langle v_1,v_2,\ldots,v_{\beta}, w^1,w^2,\ldots,w^{\alpha} \rangle$$
    It is easy to observe that each resident's list is derived from this master list.

    To show that hospitals' preference lists are derived from a master list, we present the following master list of residents. Let $A$ denote an arbitrary strict ordering of $a$-residents $a_1,a_2,\ldots,a_{\beta}$. Let $B_p$ denote the tie containing all (at most three) $b_p$-residents: $b_p^{i_1}, b_p^{i_2}, b_p^{i_3}$ for some $1\le i_1,i_2,i_3\le \alpha$.  Let $B$ denote the arbitrary strict ordering of $B_1, B_2,\ldots, B_\beta$. Let $D^s$ denote the tie containing $d_1^s$ and $d_2^s$. Let $D$ denote an arbitrary strict ordering of $D_1, D_2,\ldots, D_\alpha$. Then the master list of residents is given as $\langle A, B, D \rangle$. It is easy to observe that each hospital's preference list is derived from this master list.
\end{proof}

This completes the proof of Theorem~\ref{thm:minsumcost}.

\section{MINMAX-SS Problem}\label{sec:MINMAXProblem}

In this section, we consider the \hyperref[prob:minmaxss]{\MINMAXSS} problem. We show that the \hyperref[prob:minmaxss]{\MINMAXSS} problem is NP-hard even for a very restricted setting.
%
For this, we consider a special case of this problem, where $q(h)=1$ for all $h\in \HH$ and a budget $\ell=1$. We refer to this special case as the \CAPSSMP, since the quota of each hospital in the augmented instance is restricted to either 1 or 2. Next, we show the hardness of \CAPSSMP.

\subsection{NP-hardness of \CAPSSMP}\label{subsec:oneOrTwoCap}

We prove the hardness of \CAPSSMP\ by reducing from an instance of the \MNESAT. Let $\alpha$ and $\beta$ be two non-negative integers.  The \MNESAT\ is a variant of the Boolean satisfiability problem where the input is a conjunction of clauses. Each clause is a disjunction of exactly three variables, and no variable appears in negated form. The goal is to determine whether there exists a truth assignment to the variables such that for each clause, at least one variable is set to true and at least one to false. This problem is known to be NP-complete~\cite{PorschenSSW2014NAE3SATdam} even when each variable appears in exactly four clauses\cite{DarmannD20NAESATres}. 

\vspace{0.1in}

\noindent\textbf{Gadget reduction.} 
Let $\II$ be an instance of the \MNESAT, where each variable appears in exactly four clauses. Let $\{X_1,X_2,\ldots,X_{\beta}\}$ be the set of variables in $\II$ and $\{C_1,C_2,\ldots,C_{\alpha}\}$ be the set of clauses in $\II$.

For convenience, we fix a cyclic ordering $\mathcal{C}$ of all clauses in $\II$. Let $C_s= (X_i \vee X_j \vee X_k)$ be a clause in $\mathcal{C}$.  Let \ct{i_1} denote the next clause in which $X_i$ appears, immediately after \ct{s} in $\mathcal{C}$. We also let \ct{i_3} denote the clause in the cyclic ordering $\mathcal{C}$ where $X_i$ appears for the fourth time, starting at \ct{s}. Similarly, we define clauses \ct{j_1}, \ct{k_1}, \ct{j_3} and \ct{k_3}. 

Given $\II$, we construct an instance $G=(\RR\cup\HH,E)$ of \CAPSSMP\ with quota $q(h)=1$ for each $h\in\HH$ such that $G$ does not admit a \SSM. We show that there exists an augmented instance $G'=(\RR\cup\HH,E)$ where the maximum increase in quota for any hospital is at most 1 such that $G'$ admits a \SSM\ if and only if there exists a satisfying assignment for $\II$ such that at least one and at most two variables in each clause are set to true.


In our reduced instance $G$, there exists a gadget \gt{s} corresponding to each clause \ct{s} in $\II$. The gadget \gt{s} consists of a resident set $\RR_s = \{a_i^s, a_j^s, a_k^s, b_i^s, b_j^s, b_k^s, d_1^s, d_2^s, d_3^s\}$ and a hospital set $\HH_s = \{v_i^s, v_j^s, v_k^s, w^s\}$. The preference lists of residents and hospitals are given in Fig.~\ref{fig:preflistsMinMax}. The preference list of a resident $b_p^s$ corresponding to the variable $X_p$ for $p\in\{i,j,k\}$ consists of three hospitals -- two within the gadget \gt{s}, and one outside the gadget \gt{s}. Specifically,  the preference list of $b_i^s$ consists of hospitals $v_i^s, v_i^{i_1}$ and $ w^s$ in this order. Analogously, the hospital $v_i^s$ corresponding to  $X_i$ ranks the resident $a_i^s$ as its top choice, followed by a tie of length two consisting of the two $b$-residents, namely $b_i^s, b_i^{i_3}$, from two different gadgets. This completes the description of our reduction.

\begin{figure}[ht]
    \begin{subfigure}{.2\textwidth}
         \begin{flushleft}
            \begin{tabular}{lll}
			$a_i^s$ &:& $v_i^s$\\[2pt]
			$a_j^s$ &:& $v_j^s$\\[2pt]
                $a_k^s$ &:& $v_k^s$\\[2pt]
			$b_i^s$ &:& $v_i^s, v_i^{i_1}, w^s$\\[2pt]
                $b_j^s$ &:& $v_j^s, v_j^{j_1}, w^s$\\[2pt]
                $b_k^s$ &:& $v_k^s, v_k^{k_1}, w^s$\\[2pt]
			$d_1^s$ &:& $w^s$\\[2pt]
                $d_2^s$ &:& $w^s$\\[2pt]
                $d_3^s$ &:& $w^s$\\[2pt]
		\end{tabular}
            \caption{}
        \label{subfig:prefAMinMAX}
        \end{flushleft}
    \end{subfigure}%
    \begin{subfigure}{.3\textwidth}
        \begin{flushright}
            \begin{tabular}{lll}
	        $v_i^s$&:& $a_i^s, (b_i^s,b_i^{i_3})$\\[2pt]
                $v_j^s$&:& $a_j^s, (b_j^s,b_j^{j_3})$\\[2pt]
                $v_k^s$&:& $a_k^s, (b_k^s,b_k^{k_3})$\\[2pt]
                $w^s$ &:&  $(b_i^s, b_j^s, b_k^s), (d_1^s,d_2^s,d_3^s)$\\[10pt]  
            \end{tabular}
            \caption{}
            \label{subfig:prefBMiNMax}
        \end{flushright}
\end{subfigure}
\caption{\label{fig:preflistsMinMax}%
(i) Preference lists of residents in the gadget \gt{s}. (ii) Preference lists of hospitals in the gadget \gt{s}.
}
\end{figure}

\vspace{0.1in}

\noindent\textbf{Correctness:} Using the similar argument as in Section~\ref{sec:costVersion}, we claim that the reduced instance $G$ does not admit a \SSM. Recall that $\II$ is an instance of \MNESAT. A satisfying assignment for an instance of \MNESAT\ is an assignment of variables such that for each clause, at least one and at most two variables are set to true.

\begin{lem}\label{lem:truthvalImpliesSSM}
    If there exists a satisfying assignment for $\II$, then there exists an instance $G'$ obtained from $G$ by setting each hospital quota either 1 or 2 such that $G'$ admits a \SSM. 
\end{lem}

\begin{proof}
    We obtain an instance $G'$ from $G$ by modifying the hospital quotas as follows. Set $q'(v_r^s) = 2$ if and only if $X_r$ in clause \ct{s} is set to true. Since the satisfying assignment has the property that each clause \ct{s} in $\II$ has at least one and at most two variables set to true, the gadget \gt{s} must have the property that at least one and at most two of $q'(v_i^s), q'(v_j^s), q'(v_k^s)$ are set to 2. Additionally, we set $q'(w^s) = 2$ if and only if clause \ct{s} is satisfied by setting exactly one of $X_i, X_j$, and $X_k$ to true in the given assignment. 
    
    Next, we construct a \SSM\ in $G'$. We consider two cases based on the number of variables in clause \ct{s} that are set to true in the satisfying assignment.
    
    \noindent\textbf{Case 1}: Exactly one variable in \ct{s} is set to true. Without loss of generality, assume that $X_i$ is true, and $X_j$ and $X_k$ are false in \ct{s}. Clearly, $q'(v_i^s)=q'(w^s)=2$ and $q'(v_j^s)=q'(v_k^s)=1$. We construct a matching $M_s$ in \gt{s} as follows:\\ $M_s=\{(a_i^s,v_i^s), (b_i^s,v_i^s),(a_j^s,v_j^s), (b_j^s,w^s), (a_k^s,v_k^s), (b_k^s,w^s)\}$.  
    
    \noindent\textbf{Case 2}: Exactly two variables in \ct{s} are set to true. Without loss of generality, assume that $X_i$ and $X_j$ are true and $X_k$ is false in \ct{s}. Clearly, $q'(v_i^s)=q'(v_j^s)=2$ and $q'(v_k^s)=q'(w^s)=1$. We construct a matching $M_s$ in \gt{s} as follows:\\  $M_s=\{(a_i^s,v_i^s), (b_i^s,v_i^s),(a_j^s,v_j^s), (b_j^s,v_j^s), (a_k^s,v_k^s), (b_k^s,w^s)\}$.  

    Now, let $M'=\bigcup_{s=1}^{\alpha} M_s$. To show the strong stability of $M'$, we prove that for any $s$, no resident in the gadget \gt{s} participates in a strong blocking pair with respect to $M'$. We consider two cases based on how $M_s$ was constructed.

    \noindent\textbf{Case (a):} $M_s$ was constructed as in Case 1. Each of $a_i^s, a_j^s, a_k^s$, and $b_i^s$ is matched with their rank-1 hospital. The residents $b_j^s$ and $b_k^s$ are matched with their rank-3 hospital $w^s$. However, the two hospitals that they prefer over $w^s$ have unit quotas (because the satisfying assignment is consistent), and they are matched to their respective rank-1 residents. The $d$-vertices in gadget \gt{s} cannot block $M'$ because the quota of $w^s$ is two, and both positions are occupied by better-preferred residents.

    \noindent\textbf{Case (b):} $M_s$ was constructed as in Case 2. Each of $a_i^s, a_j^s, a_k^s, b_i^s$ and $b_j^s$ is matched with its rank-1 hospital. The resident $b_k^s$ is matched with its rank-3 hospital $w^s$. However, the two hospitals $v_k^s$ and $v_k^{k_1}$, which $b_k^s$ prefers over $w^s$, have unit quotas (because the satisfying assignment is consistent), and they are matched to their respective rank-1 residents. The $d$-vertices in gadget \gt{s} cannot block $M'$ because the unit quota of $w^s$ is occupied by a better-preferred resident. Thus, $M'$ is a \SSM.
\qed\end{proof}

We require the following claims to prove the other direction.

\begin{cl}\label{cl:atLeast1vDuplicated}
    Let $G'$ be an instance obtained from $G$ with $q'(h)\in \{1,2\}$ for all $h\in \HH$. Assume that $G'$ admits a \SSM\ $M'$. Then, for any gadget \gt{s}, $q'(h)=2$ for some $h\in\{v_i^s,v_j^s,v_k^s\}$.
\end{cl}

\begin{proof}
     Suppose for contradiction that $q'(h)=1$ for all $h\in\{v_i^s,v_j^s,v_k^s\}$. We note that $(a_p^s,v_p^s)\in M'$ for all $p\in\{i,j,k\}$; otherwise, $(a_p^s,v_p^s)$ would block $M'$. Also, note that all three $b_p^s$ for $p\in\{i,j,k\}$ are matched in $M'$, as otherwise  $(b_p^s,w^s)$ would block $M'$. Clearly, none of the $b$-vertices in \gt{s} can be matched to their rank-1 hospitals in $M'$. We claim that $b_p^s$ for $p\in\{i,j,k\}$ cannot be matched to its rank-2 hospital in $M'$. To see this, without loss of generality, let us assume that $b_i^s$ is matched with its rank-2 hospital $v_i^{i_1}$, that is, $(b_i^s,v_i^{i_1})\in M'$. The fact that $(a_i^{i_1},v_i^{i_1})\in M'$ and $q'(v_i^{i_1})\le 2$ implies that the resident $b_i^{i_1}$ is not matched with $v_i^{i_1}$. This implies that $(b_i^{i_1},v_i^{i_1})$ blocks the matching $M'$. Therefore, it must be the case that all three of $b_i^s,b_j^s$, and $b_k^s$ are matched with their rank-3 neighbor $w^s$. But, $q'(w^s)\le 2$, the hospital $w^s$ can accommodate at most two of these $b$-vertices. Thus, the unmatched resident, say $b_i^s$, along with $w^s$ blocks $M'$. 
\qed\end{proof}

\begin{cl}\label{cl:atmost2vDuplicated}
    Let $G'$ be an instance obtained from $G$ with $q'(h)\in \{1,2\}$ for all $h\in \HH$. Assume that $G'$ admits a \SSM\ $M'$. Then, $q'(h)=1$ for some $h\in\{v_i^s,v_j^s,v_k^s\}$.
\end{cl}

\begin{proof}
    Suppose for contradiction that $q'(h)=2$ for all $h\in\{v_i^s,v_j^s,v_k^s\}$. This implies that $(b_p^s,v_p^s)\in M'$ for all $p\in\{i,j,k\}$; otherwise, $(b_p^s,v_p^s)$ would block the matching $M'$. This implies that $w^s$ is not matched to any of its rank-1 residents in $M'$. Since $q'(w^s)\le 2$, an unmatched $d$-vertex in \gt{s}, say $d_1^s$, along with $w^s$ blocks $M'$.
\qed\end{proof}

\begin{cl}\label{cl:allVDuplicated}
    Let $G'$ be an instance obtained from $G$ with $q'(h)\in \{1,2\}$ for all $h\in \HH$. Assume that $G'$ admits a \SSM\ $M'$. For any gadget \gt{s}, if $q'(v_i^s)=2$, then $q'(v_i^{i_3})=2$.
\end{cl}
\begin{proof}
    Clearly, $(a_i^s,v_i^s)\in M'$. Since $q'(v_i^s)=2$, it must be the case that $(b_i^s,v_i^s)\in M'$, as otherwise, $(b_i^s,v_i^s)$ blocks $M'$. This implies $|M'(v_i^s)|=2=q'(v_i^s)$, and therefore, both positions of $v_i^s$ are occupied by the residents in \gt{s}. Since $v_i^s$ is matched to a rank-2 resident, the other rank-2 resident $b_i^{i_3}$ of $v_i^s$ which belongs to a different gadget $G_{i_1}$ must be such that $M'(b_i^{i_3})$ is within rank-2 in \prefbi. The fact that both positions of $v_i^s$ are occupied by the residents in \gt{s} implies that  $M'(b_i^{i_3})= v_i^{i_3}$. Also, $M'(a_i^{i_3})=v_i^{i_3}$, as otherwise, $(a_i^{i_3},v_i^{i_3})$ blocks $M'$.  This implies $q'(v_i^{i_3})=2$.
\qed\end{proof}

\begin{lem}\label{lem:SSMimpliesTruthVal}
    If there exists an instance $G'$ obtained from $G$ with $q'(h)\in \{1,2\}$ for all $h\in \HH$ such that $G'$ admits a \SSM, say $M'$, then the instance $\II$ admits a satisfying assignment such that at least one and at most two variables in each clause are set to true.
\end{lem}
\begin{proof}
    We obtain a truth assignment for variables in $\II$ using the quotas of hospitals in $G'$. Set $X_i$ in clause \ct{s} to true if $q'(v_i^s)=2$. Otherwise, set $X_i$ in clause \ct{s} to false. By using Claim~\ref{cl:atLeast1vDuplicated} and Claim~\ref{cl:atmost2vDuplicated}, we know that each clause has at least one and at most two variables set to true. This implies that every clause is satisfied. Using Claim~\ref{cl:allVDuplicated}, we conclude that this assignment is consistent. Therefore, every clause has at least one and at most two variables set to true. Thus, we have the required satisfying assignment for~$\II$.
\qed\end{proof}

Next, we show that the preference lists of residents admit single-peakedness\footnote{Single-peaked preferences originate from the study of elections, where preferences are typically assumed to be strict and complete. Preferences are called single-peaked if there is an arrangement of alternatives such that each voter’s preference graph has only one local maximum. In this graph, alternatives are on the $x$-axis, and scores assigned by voters are on the $y$-axis. Single-peaked preferences have also been considered for stable matchings~\cite{bartholdi1986stable,alcalde1994top,bredereck2020stable}. Social choice literature has adapted these notions to incomplete preferences as well~\cite{bredereck2020stable,fitzsimmons2020incomplete}. For the bipartite graph $G=(\RR\cup\HH,E)$, providing a strict linear ordering over $\RR$ and $\HH$ separately suffices.} property. 
We note that the preference list of a degree-2 vertex is always single-peaked with respect to any linear ordering, and the preference list of a degree-3 vertex is single-peaked with respect to the strict linear ordering if its least-preferred neighbor does not lie between its two other neighbors in that ordering. Now, we give the strict linear ordering of hospitals in $G$ as follows. Let $V^s$ denote the strict order $\langle v_i^s, v_j^s,v_k^s \rangle$ and $W$ to denote the strict order $\langle w^1,w^2,\ldots,w^{\alpha}\rangle$. Then the strict linear ordering of hospitals is given as $\langle V^1, V^2, \ldots, V^{\alpha}, W\rangle$. It is trivial to observe that for any resident $r$, the preference list of $r$ has a single peak over the above linear ordering. Thus, resident preferences are single-peaked. 

Now, we show that the hospital preference lists are derived from a master list. We give a master list of residents as follows. Let $A^s$ denote the strict order $\langle a_i^s, a_j^s,a_k^s \rangle$, $B$ denote the tie containing all the $b$-residents and $D^s$ denote the tie  $( d_1^s, d_2^s,d_3^s)$. Then, the master list of residents is given as $\langle A^1, A^2, \ldots, A^{\alpha}, B, D^1, D^2,\ldots, D^{\alpha}\rangle$. It can be easily verified that the preference list of each hospital is derived from this master list.

Lemma~\ref{lem:truthvalImpliesSSM} and Lemma~\ref{lem:SSMimpliesTruthVal}, together with the above discussion about single-peakedness and master lists, complete the proof of the first part of Theorem~\ref{thm:minmax}. 

\vspace{0.1in}

\noindent\textbf{Resident-perfect \hyperref[prob:minmaxss]{\MINMAXSS} problem:} This is a variant of the \hyperref[prob:minmaxss]{\MINMAXSS} problem, which asks for resident-perfect \SSM. We now give a simple modification in the above reduction to show that the resident-perfect \hyperref[prob:minmaxss]{\MINMAXSS} problem is also NP-hard.

In the above reduction, we introduce a unique hospital corresponding to each of $d_1^s, d_2^s$ and $d_3^s$ for $1\le s\le \alpha$, say $f_1^s,f_2^s$ and $f_3^s$, respectively, such that $d_p^s$ ranks $f_p^s$ for $1\le p\le 3$ at rank 2. These $f$-hospitals act as last resorts for each of $d$-residents. Now, it is easy to see that there exists a valid satisfying assignment for $\II$ if and only if there exists an instance $G'=(\RR\cup\HH, E)$ with $q'(h)\in\{1,2\}$ for all $h\in \HH$  obtained from $G$ such that $G'$ admits a resident-perfect \SSM.

This completes the proof of the remaining part of Theorem~\ref{thm:minmax}.

\subsection{MINMAX-SS problem with bounded ties}\label{subsec:MINMAXbounded}
In this section, we consider the \hyperref[prob:minmaxssbt]{\MINMAXSSBT} problem, which is a special case of the \hyperref[prob:minmaxss]{\MINMAXSS} problem. 
First, we observe that executing Algorithm~\ref{algo:MinSum} on an instance $G$ of the \hyperref[prob:minmaxssbt]{\MINMAXSSBT} problem produces an $\ell$-augmented instance $G'$. This follows because the lengths of ties are bounded by $\ell+1$: a hospital $h$ proposes to residents at rank $k$ only if $h$ remains under-subscribed after having proposed to all residents at ranks $1$ through $k-1$. The facts that $(i)$ $h$ was under-subscribed before proposing to $k^{th}$-rank residents, and $(ii)$ the length of the tie at rank $k$ is at most $\ell+1$, together imply that hospital $h$ cannot become over-subscribed by more than $\ell$ residents. Thus, for a given \hyperref[prob:minmaxssbt]{\MINMAXSSBT} instance, the existence of an $\ell$-augmented instance is guaranteed and can be computed efficiently. 

However, the \SSM\ $M'$ returned by Algorithm~\ref{algo:MinSum} need not be resident-optimal. For example, for the instance shown in Fig.~\ref{subfig:prefGBT1}, executing Algorithm~\ref{algo:MinSum} on $G$ outputs the augmented instance $G'$ with $q'(h_1)=2$ and $q'(h_2)=1$ and the \SSM\  $M' = \{(r_1, h_2), (r_2, h_1), (r_3, h_1)\}$. The resident-optimal \SSM\ in $G'$ is $M'' = \{(r_1, h_1), (r_2, h_2), (r_3, h_1)\}$, which is better than $M'$ for $r_1, r_2$ and no worse for the remaining residents. 

\begin{figure}
\floatbox[{\capbeside\thisfloatsetup{capbesideposition={right,center},capbesidewidth=8.5cm}}]{figure}[\FBwidth]
{\caption{}\label{fig:Clonedgraph}}
{\begin{tabular}{lll p{0.3 cm} | p{0.3 cm} lll}
			$r_1$ &:& $h_1,h_2$ & & & & &\\[2pt]
			$r_2$ &:& $h_2,h_1$ & & & & &\\[2pt]
                $r_3$ &:& $h_1$& & & $[1]\ h_1$&:& $(r_2,r_3), r_1$\\[2pt]
                $r_4$ &:& $h_2$& & & $[1]\ h_2$&:& $r_1,r_2,r_4$\\[2pt]
		\end{tabular}
            \caption{The instance $G$ with $\ell = 1$ does not admit a \SSM. Executing Algorithm~\ref{algo:MinSum} on $G$ outputs an augmented instance $G'$ with $q'(h_1) = 2$ and $q'(h_2) = 1$, and a \SSM\ $M' = \{(r_1, h_2), (r_2, h_1), (r_3, h_1)\}$. However, $M'$ is not the resident-optimal \SSM\ in $G'$. Another augmented instance 
            $\widehat{G}$ obtained by setting $\hat{q}(h_1)=\hat{q}(h_2)=2$ contains the matching $\widehat{M}=\{(r_1,h_1), (r_2,h_2), (r_3,h_1), (r_4,h_2)\}$, which is a resident-optimal \SSM\ across all augmentations.}
        \label{subfig:prefGBT1}}
\end{figure}

A natural remedy is to first compute the $\ell$-augmented instance $G'$ by executing Algorithm~\ref{algo:MinSum} on the given instance $G$, and then run the resident-proposing algorithm of Irving~\etal~\cite{irving2003strong} on $G'$. This guarantees a resident-optimal \SSM\ in the $\ell$-augmented instance $G'$, because it is known that the $\RR$-proposing algorithm outputs a resident-optimal \SSM\ whenever the given instance admits one~\cite{manlove1999stable}.  If we do so on the instance shown in Fig.~\ref{subfig:prefGBT1}, we get the matching $M''$ which is indeed resident-optimal in $G'$. However, $G'$ is just one of the many $\ell$-augmented instances of $G$ that are possible, and different augmentations can lead to different resident-optimal matchings in the respective augmented instances. In fact, the matching $M''$ is not the best for the residents among \SSM s across all $\ell$-augmented instances of $G$. 

For example, the instance $\widehat{G}$ obtained from $G$ (in Fig.~\ref{subfig:prefGBT1}) by augmenting the capacities of both $h_1$ and $h_2$ by one, \textit{i.e.} $\hat{q}(h_1)=\hat{q}(h_2)=2$ admits a matching $\widehat{M}=\{(r_1,h_1), (r_2,h_2), (r_3,h_1), (r_4,h_2)\}$, which is a \SSM\ in $\widehat{G}$. Moreover, $\widehat{M}(r)\succeq_r M''(r)$ for each resident $r$, and $\widehat{M}(r_4)\succ_{r_4} M''(r_4)$. 

This example demonstrates that simply executing the resident-proposing algorithm by \cite{irving2003strong} on the augmented instance $G'$ output by Algorithm~\ref{algo:MinSum} does not necessarily yield the best possible outcome for the residents. In this context, we define and consider a \textit{resident-optimal $\ell$-augmented} instance of $G$ (see Section~\ref{sec:problems} for definition).


\vspace{0.1in}
 \noindent\textbf{Our algorithm}
\vspace{0.1in}

 \noindent We present an efficient algorithm that outputs a resident-optimal $\ell$-augmented instance $G'$. We begin by setting a temporary quota for each hospital $h$ as $q_t(h)=q(h)+\ell$. Let us denote this instance by $G_t$. Next, we execute the algorithm by Irving~\etal~\cite{irving2003strong} (restricted to ties on the $\HH$-side as described in Appendix~\ref{sec:irving}) on the instance $G_t$. Let $M'$ be the tentative matching (also called the engagement graph) obtained at the end of the execution by the algorithm of Irving~\etal~\cite{irving2003strong}. 
 We note that a hospital $h$ may be oversubscribed in $M'$ with respect to its original quota $q(h)$. However, for every hospital $h$, we have $|M'(h)| \le q_t(h)$. 
 We further note that the matching $M'$ need not be strongly stable in $G_t$. Yet, we show that it is possible to fix the quota of each hospital $h$  to $q'(h)=\max\{q(h),|M'(h)|\}$ such that the tentative matching $M'$ becomes a strongly stable matching in the modified instance~$G'$. 

 The algorithm (pseudo-code given in Algorithm~\ref{algo:MinMax}) starts by setting a temporary quota $q'(h)=q(h)+\ell$ for each hospital $h\in \HH$, and by initializing a matching $M'$, where each resident is matched to $\bot$. Line~\ref{algoMinMax:while} to Line~\ref{algoMinMax:whileEnd} represents the algorithm by Irving~\etal~\cite{irving2003strong} (restricted to $\HH$-side ties). The tentative matching $M'$ at the end of the while loop is considered.  Finally, our algorithm fixes the quota $q'(h)$ for each hospital $h$ based on the matching $M'$ (see Line~\ref{algoMinMax:capUpdate}).

\begin{algorithm}
    \caption{Algorithm for \hyperref[prob:minmaxssbt]{\MINMAXSSBT} }\label{algo:MinMax}
    \DontPrintSemicolon
    \SetAlgoLined
    set $q'(h)=q(h)+\ell$ for all $h\in \HH$\label{algoMinMax:tempQ}\;
    $M'$ = $\{ (r, \bot )  \ \  \ | \ \ \  \mbox{for all $r \in \RR$ }\}$\;
    \While{ $\exists$ unmatched $r\in\RR$ such that \prefr is not empty\label{algoMinMax:while}}{
        $r$ proposes to the top-ranked hospital $h$ in \prefr\;
       $M' = M'  \cup \{(r, h)\} $\label{algoMinMax:accept}\;

       \If{$|M'(h)|> q'(h)$ \label{algoMinMax:oversub}}{
            let $r_1$ be a least-preferred resident in $M'(h)$\label{algoMinMax:reject1}\;
            suppose $r_1$ is at rank $p$ in \prefh\;
            \For{each $r_2$ at rank $p$ in \prefh}{
                \textbf{if}{$(r_2,h)\in M'$} \textbf{then} $M'=M'\setminus (r_2,h)$\label{algoMinMax:reject2}\;
            }
            \For{each $r_2$ at rank $\ge p$ in \prefh}{
                delete $h$ from \prefrT\ and $r_2$ from \prefh\label{algoMinMax:delete}\;
            }
        }
    \label{algoMinMax:whileEnd}}
    set $q'(h)=\max\{q(h),|M'(h)|\}$ for all $h\in \HH$  \label{algoMinMax:capUpdate}\;
    \Return $M'$ and $G'=(\RR\cup\HH,E)$ with $q'(h)$ for all $h\in\HH$\label{algoMinMax:return}\;
\end{algorithm}

Next, we prove the correctness of our algorithm. In the following lemma, we show that the $\max_{h\in\HH}\{q'(h)-q(h)\} \le \ell$. 
\begin{lem}\label{lem:vaildL}
    The instance $G'=(\RR\cup\HH,E)$ output by the Algorithm~\ref{algo:MinMax} has a property that $q'(h)\le q(h)+\ell$ for all $h\in \HH$.
\end{lem}
\begin{proof}
     Lines~\ref{algoMinMax:oversub}--\ref{algoMinMax:delete} of Algorithm~\ref{algo:MinMax} ensure that  $|M'(h)|\le q(h)+\ell$ for all $h\in\HH$. By definition of $q'(h)$ at Line~\ref{algoMinMax:capUpdate}, the lemma holds.
\qed\end{proof}

\begin{cl}\label{cl:atLeastQ1}
    If a hospital $h$ received at least $q(h)$ many proposals during the execution of Algorithm~\ref{algo:MinMax}, then $|M'(h)|\ge q(h)$.
\end{cl}
\begin{proof}
    Suppose, for contradiction, that $h$ received at least $q(h)$ many proposals during the execution of Algorithm~\ref{algo:MinMax}, but $|M'(h)| < q(h)$. This implies that $h$ must have rejected some of the proposals it received. A hospital rejects a resident only when it becomes over-subscribed. We note that the capacity of $h$ during the execution of the while loop of Algorithm~\ref{algo:MinMax} was $q'(h)=q(h)+\ell$. Thus, $h$ must have become over-subscribed with respect to the quota $q'(h)$ at some point. This further implies that $h$ received at least $q(h) + \ell + 1$ proposals.  By the design of the algorithm, $h$ only rejects proposals at a particular rank (the least-preferred one) at any given time. Once all the proposals at the least-preferred rank are rejected, $h$ is no longer over-subscribed. This is because, as soon as $h$ becomes over-subscribed, it rejects these least-preferred proposals. Since $|M'(h)| < q(h)$, $h$ must have rejected at least $\ell + 2$ proposals at once, all from the same rank. This implies that there exists a tie of length at least $\ell + 2$, contradicting the assumption that the maximum length of ties in the preference list of $h$ is at most $\ell + 1$.
\qed\end{proof}

We use the following claim to prove Lemma~\ref{lem:Gssm}.
\begin{cl}\label{cl:fullH}
    If $q'(h)\ge q(h)$ after the termination of Algorithm~\ref{algo:MinMax}, then $h$ must be fully subscribed in $M'$ with respect to its augmented quota $q'(h)$.
\end{cl}
\begin{proof}
    By definition of $q'(h)$ at Line~\ref{algoMinMax:capUpdate}, if $|M'(h)|\ge q(h)$, then $q'(h)=|M'(h)|$.
\qed\end{proof}

\begin{lem}\label{lem:Gssm}
    Let $G'=(\RR\cup\HH,E)$ be the instance and $M'$ be the matching produced by Algorithm~\ref{algo:MinMax}. Then the matching $M'$ is a \SSM\ in the instance $G'$.
\end{lem}
\begin{proof}
     Suppose for contradiction that $M'$ is not a \SSM\ in $G'$. Let the pair $(r,h)$ block the matching $M'$ in $G'$. Therefore, $r$ is either unmatched in $M'$ or $h\succ_r M'(r)=h'$. Clearly, $r$ proposed to $h$ during the execution of Algorithm~\ref{algo:MinMax}. Then it must be the case that either $r$ got rejected by $h$ or $h$ deleted $r$ from $\prefh$. This can happen only if $h$ received at least $q(h)+\ell+1 > q(h)$ proposals. By Claim~\ref{cl:atLeastQ1}, $|M'(h)|\ge q(h)$. Thus, in $G'$, the quota $q'(h)$ of $h$ must be such that $q'(h)\ge q(h)$. By Claim~\ref{cl:fullH}, $h$ is fully subscribed in $M'$ with respect to $q'(h)$ in the instance $G'$.
     Also, note that $h$ rejects a resident only when $r$ is the least-preferred matched resident to $h$. At the time of this rejection, $h$ also deletes all the residents that are at the same or greater rank compared to $r$ in \prefh. Thus, for all $r'\in M'(h)$, it must be the case that $r'\succ_h r$.

    Therefore, $h$ is fully subscribed in $M'$ with all residents in $M'(h)$ strictly better-preferred than $r$ in the instance $G'$, and hence the edge $(r,h)$ cannot block $M'$, contradicting the assumption that it is a blocking pair.
\qed\end{proof}

Lemma~\ref{lem:vaildL} and Lemma~\ref{lem:Gssm} together imply that the instance $G'$ produced by Algorithm~\ref{algo:MinMax} is indeed an $\ell$-augmented instance. 

Next, we prove that the matching $M'$ returned by Algorithm~\ref{algo:MinMax} is the best \SSM\ for residents across \SSM s of all $\ell$-augmented instances for the given \hyperref[prob:minmaxssbt]{\MINMAXSSBT} instance.

\begin{lem}\label{lem:BestForRes}
    The instance $G'$ returned by Algorithm~\ref{algo:MinMax} is a resident-optimal $\ell$-augmented instance of $G$. Moreover, the matching $M'$ in $G'$ is a resident-optimal \SSM\ across all $\ell$-augmented instances of $G$.
\end{lem}

\begin{proof}
    Suppose for contradiction that there exists an $\ell$-augmented instance $\widehat{G}$ and a \SSM\ $\widehat{M}$ in $\widehat{G}$ such that $\widehat{M}(r) \succ_r M'(r)$ for some resident $r\in \RR$.

    Consider the first proposal, in the execution of Algorithm~\ref{algo:MinMax}, from some resident $r$ to some hospital $h$ such that $(r,h)\in \widehat{M}$ but $h$ rejected $r$ during that execution. The existence of such a proposal is guaranteed because of our assumption that $\widehat{M}(r) \succ_r M'(r)$. The fact that $h$ rejected $r$ during our execution sequence implies that at the time when $h$ rejected $r$, it must be the case that $|M'(h)|> q'(h)=q(h)+\ell$, and each resident $r'$ in $M'(h)$ at this time must be such that $r'\succeq_h r$. Let us denote by $\RR_h^r$ all the residents in $M'(h)$ at this time. The fact that $|\RR_h^r|>q(h)+ \ell$ implies that $\widehat{M}(h) \neq \RR_h^r$. Thus, there exists $r_1\in \RR_h^r$ such that $r_1\notin \widehat{M}(h)$. If $\widehat{M}(r_1) \prec_{r_1} h$, then $(r_1,h)$ forms a \SBP\ with respect to $\widehat{M}$. Thus, $\widehat{M}(r_1) \succ_{r_1} h$. Assume that $\widehat{M}(r_1) = h_1$. 

    Since $r_1\in \RR_h^r$, it must be that $r_1$ proposed to $h$. This further implies that $r_1$ proposed to $h_1$ and $h_1$ rejected $r_1$. Clearly, the rejection of $r_1$ by $h_1$ happened before $r_1$ proposed to $h$. Note that at the time when $h$ rejected $r$, the resident $r_1$ was in $M'(h)$. Thus, rejection of $r$ by $h$ happened only after $h_1$ rejected $r_1$. This contradicts the assumption that $r$ to $h$ was the first proposal that was rejected by any hospital and $(r,h)\in \widehat{M}$.\qed
\end{proof}

Lemma~\ref{lem:BestForRes} immediately proves the following claim.

\begin{cl}\label{cl:minmaxBTcorr}
    The \SSM\ $M'$ returned by Algorithm~\ref{algo:MinMax} matches the maximum number of residents across \SSM s of all $\ell$-augmented instances of $G$.
\end{cl}

Lemma~\ref{lem:BestForRes} and Claim~\ref{cl:minmaxBTcorr} together complete the proof of Theorem~\ref{thm:MINMAXBT}

\section{Conclusion}\label{sec:Conclusion}

In this paper, we study the capacity augmentation problem for strongly stable matchings in the \HRHT\ setting. 
We investigate two natural optimization criteria, \MINSUM\ and \MINMAX.


The NP-hardness of the \MINMAXSS\ problem naturally raises the question of whether a good approximation algorithm exists; we leave this as an open question.  Our work considers only quota \textit{augmentations}; the complementary problem of quota \textit{reductions}, studied for strict preferences by Gokhale~\etal~\cite{GokhaleSNV24AAMASCap}, remains largely unexplored in the presence of ties and is a natural next step.  Finally, the \MINSUMCOST\ inapproximability result relies on zero-one costs. It would be interesting to identify structural restrictions on the cost function or preference lists under which the problem becomes tractable. For instance, uniform costs reduce to \MINSUMSS\ and are polynomial-time solvable.

%
%
%
\newpage
\bibliographystyle{splncs} 
\bibliography{refs}
\newpage
\appendix

\section{Background: An Algorithm for Strong Stability in \HRHT}\label{sec:irving}

In this section, we look at the algorithm for determining whether a given instance admits a \SSM, and for computing one if it does. Irving \etal~\cite{irving2003strong} explored strong stability in the context of the \HRT\ problem and presented an algorithm for finding a \SSM\ when one exists. Our work focuses on the \HRHT\ problem, where ties are restricted only to the hospitals' side. Therefore, we describe a simplified version of the algorithm by Irving \etal~\cite{irving2003strong}.

 The pseudo-code for the simplified version of their algorithm is presented in Algorithm~\ref{algo:HRHTSSM}. The algorithm begins with every resident unmatched, that is, each resident is initially matched to $\bot$ (see Line~\ref{algoHRHTSSM:InitMatch}). For each hospital $h$, the algorithm maintains a variable $full(h)$, to track whether $h$ becomes fully subscribed during execution. Initially, $full(h) = \text{false}$ for every $h \in \HH$ (see Line~\ref{algoHRHTSSM:notFull}). During the algorithm, each unmatched resident $r$ proposes to its top-ranked hospital, say $h$.  When $h$ receives a proposal from $r$, it provisionally accepts it (see Line~\ref{algoHRHTSSM:accept}). If $h$ becomes fully subscribed, $full(h)$ is set to true. In this process, some hospitals may become over-subscribed.  If a hospital $h$ is over-subscribed, it rejects all residents matched to it in $M$ at the least-preferred rank (see Lines~\ref{algoHRHTSSM:reject1}--\ref{algoHRHTSSM:reject2}). This is because none of the residents at this rank or lower-preferred rank can be matched to $h$ in any \SSM. Therefore, all such edges are deleted from the given instance (see Line~\ref{algoHRHTSSM:delete}). Subsequently, for each resident, $r\in \RR$, either $r$ is matched to some hospital or \prefr\ becomes empty. The proposal sequence terminates when either all residents are matched or \prefr\ for each unmatched resident becomes empty.

 After the proposal sequence terminates, the algorithm uses the tracking variable $full(h)$ to determine whether a \SSM\ exists. If there is a hospital $h$ such that $full(h)=$ true and $h$ is under-subscribed in the final matching $M$, then $M$ is blocked by $(r',h)$ where $r'$ is a resident who was rejected by $h$ during the execution of Algorithm~\ref{algo:HRHTSSM}. In this case, Algorithm~\ref{algo:HRHTSSM} concludes that no \SSM\ exists for the given instance. On the other hand, if for every $h$ such that $full(h)=$ true, $h$ is not under-subscribed, then Algorithm~\ref{algo:HRHTSSM} declares that $M$ is a \SSM\ for the given instance.

\begin{algorithm}[h]
    \caption{Algorithm for \SSM\ in \HRHT\ Problem }\label{algo:HRHTSSM}
    \DontPrintSemicolon
    \SetAlgoLined
    $M$ = $\{ (r, \bot )  \ \  \ | \ \ \  \mbox{for every resident $r \in \RR$ }\}$\label{algoHRHTSSM:InitMatch}\;
    set $full(h)= \text{false}$ for each hospital $h\in \HH$\label{algoHRHTSSM:notFull}\;
    \While{$\exists$ unmatched $r\in\RR$ such that \prefr is not empty\label{algoHRHTSSM:while}}{
        $r$ proposes to the top-ranked hospital $h$ in \prefr\;
       $M = M  \cup \{(r, h)\} $\label{algoHRHTSSM:accept}\;
       \If{$|M(h)|\ge q(h)$ \label{algoHRHTSSM:fullysub}}{
            $full(h)=\text{true}$\;
       \If{$|M(h)|> q(h)$ \label{algoHRHTSSM:oversub}}{
            let $r'$ be a least-preferred resident in $M(h)$\label{algoHRHTSSM:reject1}\;
            suppose $r'$ is at rank $p$ in \prefh\;
            \For{each $r''$ at rank $p$ in \prefh}{
                \textbf{if}{$(r'',h)\in M$} \textbf{then} $M=M\setminus (r'',h)$\label{algoHRHTSSM:reject2}\;
            }
            \For{each $r''$ at rank $\ge p$ in \prefh}{
                delete $h$ from \prefrpp\ and $r''$ from \prefh\label{algoHRHTSSM:delete}\;
            }
        }
    }
    }
    \If{$\exists\ h\in \HH$ such that $full(h)=$ \text{true} and $|M(h)|<q(h)$\label{algoHRHRTSSM:check}}{
        $G$ does not admit a \SSM\label{algoHRHTSSM:NoInstance}\;
    }
    \Else{
        \Return $M$ as a \SSM\ in $G$\label{algoHRHTSSM:return}\;
    }
    
\end{algorithm}

Now, we prove the correctness of this algorithm. First, we show that if Algorithm~\ref{algo:HRHTSSM} returns a matching $M$, then $M$ is indeed a \SSM. We note that by the design of the algorithm, no hospital can be over-subscribed in $M$. This is because of the following reason. Since preference lists of residents are strict and at any given time only one resident can make a proposal, a hospital $h$ can become over-subscribed by at most one position.  As soon as a hospital becomes over-subscribed, Line~\ref{algoHRHTSSM:oversub} -- Line~\ref{algoHRHTSSM:delete} are executed, and $h$ rejects at least one resident. Consequently, no hospital remains over-subscribed, ensuring that $M$ is a valid matching. Now, we show the strong stability of $M$.

\begin{lem}\label{lem:algoHRHT1}
    The matching $M$ returned by Algorithm~\ref{algo:HRHTSSM} 
    is a \SSM.
\end{lem}
\begin{proof}
    Suppose, for contradiction, that $M$ is not a \SSM. Then there must exist a pair $(r,h)$ that blocks $M$. We first show that $(r,h)$ was not deleted during the execution of Algorithm~\ref{algo:HRHTSSM}. Suppose that the pair $(r,h)$ was deleted by Algorithm~\ref{algo:HRHTSSM}. We consider two cases: (i) $h$ is fully subscribed in $M$: In this case, by the design of the algorithm, all residents matched to $h$ in $M$ are strictly better-preferred over $r$ -- contradicting the assumption that $(r,h)$ blocks $M$. (ii) $h$ is under-subscribed in $M$: Since $h$ rejected $r$, it must be the case that $full(h)=$ true. However, if $full(h)=$ true and $h$ is under-subscribed in $M$, then Algorithm~\ref{algo:HRHTSSM} would have concluded that $G$ does not admit a \SSM\ -- a contradiction. Therefore, the algorithm could not have deleted the pair $(r,h)$.

    The fact that $(r,h)$ was not deleted implies that \prefr\ did not become empty, and hence, $r$ is matched in $M$. Let $M(r)=h'$. Since $(r,h)$ blocks $M$, and resident preference lists are strict, it must be the case that $h\succ_r h'$. The fact that $r$ proposed to $h'$ implies that $h$ rejected $r$ and hence the pair $(r,h)$ was deleted by the Algorithm~\ref{algo:HRHTSSM} -- a contradiction.
\qed\end{proof}

 For the rest of the proof, we assume that $M$ denotes the matching obtained when the proposal sequence terminates, that is, just before reaching Line~\ref{algoHRHRTSSM:check}. We prove the following claims, which will be useful in proving the correctness of Line~\ref{algoHRHTSSM:NoInstance} of the Algorithm~\ref{algo:HRHTSSM}. 

\begin{cl}\label{cl:algoHRHTSSMPair}
    Suppose Algorithm~\ref{algo:HRHTSSM} deleted a pair $(r,h)$ during the execution. Then the pair $(r,h)$ cannot belong to any \SSM.
\end{cl}
\begin{proof}
    Suppose for contradiction that $(r,h)$ was deleted by Algorithm~\ref{algo:HRHTSSM}, but $(r,h)\in \tM$ for some \SSM\ $\tM$. Without loss of generality, assume that $(r,h)$ was the first strongly stable pair that was deleted during the execution of Algorithm~\ref{algo:HRHTSSM}. We note that $h$ is not over-subscribed in $\tM$. Suppose the pair $(r,h)$ was deleted at time $t$. This implies that $h$ was over-subscribed with residents, say $M^t(h)$ , at time $t$. This further implies that for each $r'\in M^t(h)$ we have that $r' \succeq_h r$.  Since $(r,h)\in \tM$, there must exist $r''\in M^t(h)$ such that $(r'',h)\notin \tM$. We claim that there does not exist any \SSM, say $M''$, such that $(r'',h'')\in M''$ for some $h'' \succ_{r''} h$. If this is the case, then the pair $(r'',h'')$ would have been deleted before $(r,h)$ -- a contradiction. This implies that $h\succ_{r''} \tM(r'')$. Since $r''\in M^t(h)$, $r'' \succeq_h r$. The fact that $(r'',h) \notin \tM$ implies that $(r'',h)$ blocks the supposed \SSM\ $\tM$ -- a contradiction. 
\qed\end{proof}

\begin{cl}\label{cl:algoHRHTUnmatched}
    If a resident $r$ is unmatched in $M$, then $r$ cannot be matched in any \SSM.
\end{cl}
\begin{proof}
    Suppose a resident $r$ is unmatched in $M$. This implies during the execution of Algorithm~\ref{algo:HRHTSSM}, \prefr\ becomes empty. That is, each pair $(r,h)\in E$ is deleted by the Algorithm~\ref{algo:HRHTSSM}. Applying Claim~\ref{cl:algoHRHTSSMPair}, we conclude that $r$ remains unmatched in each \SSM. 
\qed\end{proof}

\begin{cl}\label{cl:algoHRHTSSMMustMatch}
    Suppose $G$ admits a \SSM\ $\tM$, and there exists a hospital $h$ such that $|\tM(h)|<q(h)$. If a resident $r$ proposed to $h$ during the course of Algorithm~\ref{algo:HRHTSSM}, then $r\in \tM(h)$.
\end{cl}
\begin{proof}
    The fact that $r$ proposed to $h$ during the execution of Algorithm~\ref{algo:HRHTSSM} implies that each pair $(r,h')$ such that $h'\succ_{r} h$ is deleted. Thus, by applying Claim~\ref{cl:algoHRHTSSMPair}, we know that $\tM(r) \not\succ_{r} h$. If $h\succ_{r}\tM(r)$, then $(r,h)$ blocks $\tM$ as $|\tM(h)|<q(h)$.  Therefore, $r\in \tM(h)$. 
\qed\end{proof}
\begin{lem}\label{lem:algoHRHT2}
    If there exists a hospital $h$ such that $full(h)=$ true, but $|M(h)|<q(h)$, then $G$ does not admit a \SSM.
\end{lem}

\begin{proof}
    Suppose for contradiction that $G$ admits a \SSM, say $\tM$. The fact (i) $full(h)=$ true, and (ii) $|M(h)|<q(h)$ together imply that there exists a resident $r^*$ such that the proposal by $r^*$ was provisionally accepted by $h$ but was later rejected by $h$. This further implies that the pair $(r^*,h)$ was deleted by Algorithm~\ref{algo:HRHTSSM}. Applying Claim~\ref{cl:algoHRHTSSMPair}, we know that $(r^*,h)\notin \tM$. Also, note that $(r^*,h')\notin \tM$ for any $h'$ such that $h'\succ_r h$. This is because all such pairs are deleted (as $r^*$ proposed to the hospital $h$).  If we show that $|\tM(h)|<q(h)$, then it immediately implies that $(r^*,h)$ blocks $\tM$. Next, we show that $|\tM(h)|<q(h)$.
    
    Let us assume that the set of residents matched in $M$ is denoted by $\RR^M$, and the set of residents matched in $\tM$ is denoted by $\RR^{\tM}$. Applying Claim~\ref{cl:algoHRHTSSMMustMatch}, we observe that if $|M(h)|=q(h)$, then $|\tM(h)|=q(h)$, and if $|M(h)|< q(h)$, then $|\tM(h)|\ge |M(h)|$. Therefore, $|\tM(h)|\ge |M(h)|$ for each $h\in \HH$. Thus, $|\RR^{\tM}| \ge |\RR^{M}|$. Now, using Claim~\ref{cl:algoHRHTUnmatched}, we observe that $|\RR^{\tM}| \le |\RR^{M}|$. Thus, $|\RR^{\tM}| = |\RR^{M}|$. Since $|\tM(h)|\ge |M(h)|$, we conclude that $|\tM(h)|= |M(h)|$ for each $h\in \HH$. Thus, $|\tM(h)|=|M(h)|<q(h)$.

    Therefore, $(r^*,h)$ blocks $\tM$.
\qed\end{proof}

Using Lemma~\ref{lem:algoHRHT1} and Lemma~\ref{lem:algoHRHT2} we have the following theorem.
\begin{thm}\label{thm:Irving}
   For a given \HRHT\ instance, the existence of \SSM\ is decidable in polynomial time. Moreover, if the instance admits a \SSM, one such matching can be computed in polynomial time.
\end{thm}
\end{document}